\documentclass[%
preprint,onecolumn,
 amsmath,amssymb,
 aps,pre,
]{revtex4-1}
\usepackage{enumitem}
\usepackage{natbib}
\usepackage{bbm}
\usepackage{subfigure}
\usepackage[colorlinks=true,linkcolor=blue,
urlcolor=blue,citecolor=blue]{hyperref}
\usepackage{amsthm}
\usepackage{soul} 
\usepackage{graphicx}
\usepackage{dcolumn}
\usepackage{bm}
\usepackage{color}
\usepackage[normalem]{ulem}

\newtheorem{theorem}{\bf Theorem}[section]
\newtheorem{lemma}{\bf Lemma}[section]

\newtheorem{corollary}{\bf Corollary}[section]
\newtheorem{remark}{\bf Remark}[section]

\begin{document}

\title{Dynamics of two moving vortices in the presence of a fixed vortex}


\author{Sreethin Sreedharan K\footnote{sreethin2@gmail.com} and Priyanka Shukla\footnote{priyanka@iitm.ac.in}}

\affiliation{ Department of Mathematics,\\
Indian Institute of Technology Madras, Chennai 600036, India
}
\date{\today}




%
%

\begin{abstract}
The dynamics of a constrained three-vortex system,
a pair of point vortices of arbitrary non-zero circulations in the velocity field of a fixed point vortex, is investigated. 
The underlying dynamical system is simplified using a coordinate transformation and categorized into two cases based on the zero and non-zero values of the constant of angular impulse. 
For each case, dynamical features of the vortex motion are studied analytically in the transformed plane to completely classify the vortex motions 
and understand the boundedness and periodicity of the inter-vortex distances.    
The theoretical predictions are also verified numerically and illustrated for various sets of initial conditions and circulations.
\end{abstract}

\pacs{Valid PACS appear here}
\maketitle
\section{Introduction}

Vortices are a major driving force behind complex fluid evolutions, such as turbulent and transitional flows~\cite{hussain_1986,jeong_hussain_1995}. Studying vortex interactions is thus essential in understanding many fluid flows~\cite{Saffman_Baker_1979}. The simplest vortex model that one can envisage is the point vortex model, which approximates a vortex as a discrete singularity of vorticity in an incompressible two-dimensional ideal fluid~\cite{HH1858,tait1867integrals}. The point vortex assumption removes all analytical complications 
associated with vortices' internal dynamics, thereby enabling one to track the motion of vortex centers efficiently. Analogous to the $n$-body problems in celestial mechanics, understanding and classifying the various trajectories exhibited by a mutually interacting system of finitely many point vortices is central to point vortex theory. 
The point vortex model also has a wide range of applications in physics, e.g.,~ Bose-Einstein condensate and quantum vortices~\cite{GPHS2018}.
For a comprehensive review of the point vortex model and its fluid dynamics applications, the reader may refer to~\cite{AH2007,NP2013}.

%




The motion of two straight parallel vortices constitutes a fundametal problem of interacting vortices~\cite{Saffman_Baker_1979}. Owing to its importance in aviation (e.g. the problem of aircraft trailing wakes~\cite{CJ2005,B2011}) and geophysical fluid dynamics 
(e.g. hetons~\cite{H1985,L1996}), it is widely discussed in the literature~\cite[see][and references therein]{S2014,L2016}. In planar geometry,
a pair of vortices has a limited number of possible motions. Depending on the sum of the two vortex circulations being zero 
or not, it is either an unbounded translation with constant straight-line speed or a circular motion around the center of vorticity with constant angular velocity~\cite{L2016}, {\color{blue}respectively}. 
In both these cases, vortices keep a constant distance between them. One may further introduce an
additional dynamical complexity by adding a fixed point vortex in the vortex pair's vicinity. By a {\it fixed} point vortex we mean here a vortex that remains at a constant location on the plane irrespective of the velocity induced by the other vortices.
In ocean dynamics, a fixed point vortex is a known model  of topography that induces a closed re-circulation zone~\cite{RK2013,RE2014} in its vicinity. Hence a vortex system consisting of both a point vortex pair and a fixed point vortex can be
considered as a basic model for a vortex pair's interaction with a topographically constrained eddy in the ocean. Since the bottom topography can significantly affect the dynamics of ocean vortices, such studies
help to gain more
insights into ocean vortices~\cite{B1997}.

On the other hand, in general, a three-vortex system~\cite{G1877,S1949,N1975,AH1979} has no closed-form solutions, and the presence of more than three point vortices leads to chaotic vortex trajectories (see, e.g.~\cite{A1983,KC1989,NP2013}). 
Imposing some form of constraints to larger vortex systems may provide the necessary analytic simplification without cutting down much on the dynamical richness of the solutions. Fixing some of the vortices to specific locations on the plane is a natural candidate in this regard, since it generalizes the class of important vortex systems that can be equivalently described as {\color{blue}a passive scalar} 
advection problem involving one or more fixed vortices~\cite{KC1985,A1989,AB2016}. For instance, the system of three vortices with zero total circulation can be put into 
{\color{blue}a passive scalar} 
advection problem involving three fixed-collinear point vortices~\cite{A1989}. 
It is also worth noticing that a fixed vortex is different from a vortex at rest as it is not required that the velocities induced by the rest of the vortices at the location of a fixed vortex add up to zero. For example, in the classical two-vortex problem, vortices cannot be at rest when vortex circulations are non-zero. Yet, one could consider the two-vortex system with one of the vortices fixed, wherein the free vortex moves in a circular path around the fixed vortex. Consequently, in general, vortex systems that include fixed vortices don't fall under the classical $n$-vortex systems but instead form a unique class of their own. 

Here, we shall address the planar vortex system consisting of two freely moving point vortices influenced by a fixed point vortex's presence in their vicinity. 
In this context, the recent studies~\cite{RK2013,RE2014,KR2018} on a vortex dipole (equal counter-rotating pair of vortices) encountering a fixed vortex
showed 
that if the constant of angular momentum [c.f.~Eqn.~\ref{cons_ang}] is zero, the vortex dipole always scatters and executes an unbounded motion. On the contrary, if this value is non-zero, then the dipole vortices' motion can be bounded for some initial conditions~\cite{RK2013,KR2018}. The analytic expression for the boundary separating the bounded and unbounded regime was also derived in terms of (i) the ratio of the circulations
of the vortex dipole to that of the fixed vortex and 
(ii) a parameter based on the ratio of the free vortices' initial positions to that of the fixed vortex~\cite{KR2018}. 
A further numerical study~\cite{RE2014} of the scalar transport, using Poincar\'e sections, revealed that the periodic oscillation of an entrapped dipole about the fixed vortex perturbs scalar motion causing a portion of scalar trajectories to manifest chaotic behaviour~\cite{RE2014}. 
In short, Refs.~\cite{RK2013,KR2018,RE2014} indicate that the strength and location of the fixed vortex determines whether or not a vortex dipole gets entrapped to its neighborhood and induce chaotic stirring therein.


The main objective of the present work is to understand the planar vortex system consisting of two freely moving point vortices 
in the presence of a fixed point vortex
using the dynamical system theory perspective without restricting to any specific circulations.
A complete classification of vortex trajectories will be carried out through phase plane analysis. 
Moreover, important insights on the special case of a vortex dipole, as shown in Refs.~\cite{RK2013,KR2018}, will also be reproduced through elegant geometrical arguments using basic dynamical system theory. 
For instance, the existence of a separatrix boundary of vortex entrapment is found to be due to the presence of a saddle equilibrium point in the phase plane. 

This paper is organized as follows. The mathematical formulation of the point vortex system is given in Sec.~\ref{sec:problem_formulation}. In Sec.~\ref{sec:analysis_results}, the model at hand is explored using dynamical system theory, and the results obtained are explained through examples numerically. In particular, the dynamical aspects of two cases, symmetric ($M=0$) and asymmetric ($M\neq 0$), are discussed in Secs~\ref{subsec:Meq0} and~\ref{subsec:Mneq0}, respectively. A few examples for each case are illustrated in Secs~\ref{subsec:Examples_Meq0} and~\ref{subsec:Examples_Mneq0}. 
The derived conclusions are discussed in Sec.~\ref{sec:conclusions}.

\section{Problem formulation}
\label{sec:problem_formulation}

We consider the three-vortex problem in two-dimensional plane $\mathbb{R}^2$ with the additional constraint that one of the point vortices is fixed at some location in the plane. Let $\Gamma_\alpha$ ($\alpha=0,1,2$) be the non-zero circulation of the $\alpha$-vortex $\mathcal{V}_\alpha$ whose coordinate function is $(x_\alpha,y_\alpha)$. Without loss of generality (WLOG), let us assume that the vortex $\mathcal{V}_0$ is fixed in the plane. Consequently, $(x_0, y_0)$ is a constant function of time $t$. For simplicity, we choose the origin to be situated at the fixed vortex. Furthermore, we align and scale the coordinate axes in such a way that the vortex $\mathcal{V}_1$ is initially 
at a unit distance 
from the origin along the positive $x$-axis, as displayed in  schematic diagram~\ref{figure_model}(a). 

\begin{figure}[!th]
\centering
\includegraphics[width=\textwidth]{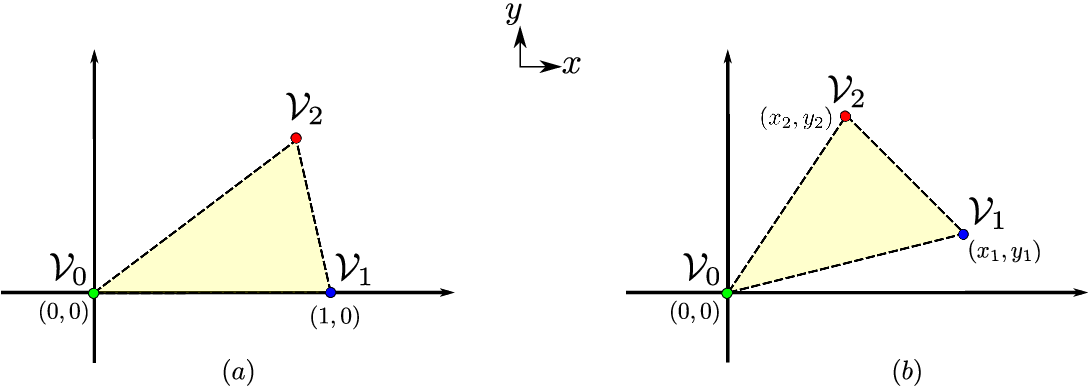}
\caption{\small Schematic showing the three-vortex model in $\mathbb{R}^2$ plane where a vortex $\mathcal{V}_0$ is fixed at $(x_0, y_0)=(0,0)$. 
(a) Initial configuration in which the vortex $\mathcal{V}_1$ is located at $(x_1,y_1)|_{t=0}=(1,0)$, and (b) configuration at time $t>0$ with free vortices $\mathcal{V}_1$ and $\mathcal{V}_2$ being located at $(x_1,y_1)$ and $(x_2,y_2)$, respectively.}
\label{figure_model}
\end{figure} 

For the sake of conciseness, from here onwards, we shall treat an element of the plane $\mathbb{R}^2$ the same as a complex number via the identification map $(x,y)\mapsto x+\mathbbm{i}\, y$. 
Each free vortex experiences the sum of the velocity field induced by the other two vortices, which gives us the following coupled system of non-linear differential equations
\begin{align}
\begin{aligned}
\dot{z_1}&= \frac{\mathbbm{i}\Gamma_2}{2\pi}\frac{z_1-z_2}{|z_1-z_2|^2}+\frac{\mathbbm{i}\Gamma_0}{2\pi}\frac{z_1}{|z_1|^2},\\
\dot{z_2}&= \frac{\mathbbm{i}\Gamma_1}{2\pi}\frac{z_2-z_1}{|z_1-z_2|^2}+\frac{\mathbbm{i}\Gamma_0}{2\pi}\frac{z_2}{|z_2|^2},
\end{aligned}
\label{z0_z1_z2_dot}
\end{align}
where $z_\alpha(t)=x_\alpha(t)+ \mathbbm{i}\, y_\alpha(t)\equiv \left(x_\alpha(t),y_\alpha(t)\right)$
is the corresponding coordinate of the vortex $\mathcal{V}_\alpha$ in the complex plane with $\mathbbm{i}$ being the imaginary unit,
%
and dot over a quantity represents its time derivative. Observe that, 
if $|z_1|=0$, $|z_2|=0$ or $|z_2-z_1|=0$, system~\eqref{z0_z1_z2_dot} is undefined. In other words, the point vortex setting fails to explain the evolution of vortices once vortex collisions are encountered during the motion. Therefore, it is necessary to assume at least initially ($t=0$) that the inter-vortex distances $r_{1}=|z_1|$, $r_2=|z_2|$, and $r_{12}=|z_2-z_1|$ are non-zero. 
Note that $r_1|_{t=0}=1$ due to the choice of coordinate axes.

To have a better geometrical understanding of vortex evolutions, we shall use polar coordinates. For $\alpha\in \{1,2\}$, let $z_\alpha(t)=r_\alpha(t)e^{i\theta_\alpha(t)}$ with $r_\alpha(t)$ and $\theta_\alpha(t)$ being the modulus and argument of $z_\alpha(t)$, respectively. 
The polar variables $(r_\alpha, \theta_\alpha)$ are related to each other by the cosine rule
\begin{equation}
r_{12}^2=r_1^2+r_2^2-2r_1r_2\cos \theta, \quad \mbox{where}\quad \theta =\theta_2-\theta_1.
\label{cosine_rule}
\end{equation}
Following a similar derivation as in Ref.~\cite{S1949}, we obtain the differential equations for inter-vortex distances
\begin{align}
r_1\dot{r_1}&=\Gamma_2 Ar_{12}^{-2}/\pi
\label{eq_R2},\\ 
r_2\dot{r_2}&=-\Gamma_1Ar_{12}^{-2}/\pi
\label{eq_R1},\\
r_{12}\dot{r_{12}}&=\Gamma_0 A(r_2^{-2}-r_1^{-2})/\pi\label{eq_R01},
\end{align}
where $A=1/2\,r_1r_2\sin\theta$ is the area of the triangle obtained by joining the three vortices $\mathcal{V}_0$, $\mathcal{V}_1$ and $\mathcal{V}_2$ with circulations $\Gamma_0$, $\Gamma_1$ and $\Gamma_2$. 
It follows from~\eqref{eq_R2}--\eqref{eq_R01} that there are two finite constants of the motion
\begin{align}
M&=\Gamma_1 r_1^2+\Gamma_2 r_2^2,
\label{cons_ang}\\
 H&=-\frac{1}{2\pi}\left(\Gamma_1\Gamma_0\log r_1+\Gamma_2\Gamma_0\log r_2+\Gamma_1\Gamma_2\log r_{12}\right),
 \label{cons_energy}
\end{align}
which arise from the conservation of angular impulse and the conservation of interaction energy of the vortex system, respectively. Indeed $H$ is nothing but the Hamiltonian of system~\eqref{z0_z1_z2_dot} satisfying Hamiltonian equations of motion for $x_\alpha$ and $y_\alpha$,
\begin{equation}
\Gamma_\alpha \dot{x}_\alpha=\frac{\partial H}{\partial y_\alpha}\quad \mbox{and}\quad \Gamma_\alpha \dot{y}_\alpha=-\frac{\partial H}{\partial x_\alpha}.
\label{eqn:canonical_EQN_Hamiltonian}
\end{equation}

Unlike the classical counterpart~\cite{G1877,S1949,N1975,AH1979}, the present problem of constrained three vortices (one fixed vortex and two free vortices) lacks the conservation of linear impulse, and the corresponding barycenter (a constant center of vorticity) symmetry associated with it. Although fewer constants of motion generally indicate non-integrability, only the variables $z_1$ and $z_2$ evolve temporally. 
Since the conserved quantities $M$ and $H$ are functionally independent, and also Poisson commute, the Hamiltonian system~\eqref{z0_z1_z2_dot} is completely integrable. 
Given that the vortex system~\eqref{z0_z1_z2_dot} 
has only regular dynamics, we may now proceed with 
understanding 
the vortex trajectories for different initial conditions.

\section{Constrained three-vortex problem}
\label{sec:analysis_results}


Motivated by the limiting case of restricted three-vortex problem wherein $\mathcal{V}_0$ is fixed at the origin and $\Gamma_2=0$~\cite{SP2019}, we define the coordinate transformation $z\mapsto z/z_1$ 
to obtain a new set of variables $\eta_0,\, \eta_1$, and $\eta_2$ defined by
\begin{equation}
\eta_0(t)=\frac{z_0}{z_1}\equiv (0,0),\quad \eta_1(t)=\frac{z_1}{z_1}\equiv (1,0) ,\quad \eta_2(t)=\frac{z_2}{z_1}.
\label{etaeqn}
\end{equation}
The transformed coordinates quantifies the relative motion of vortices with respect to the free vortex $\mathcal{V}_1$ in the complex plane. 
 There are advantages of using transformed coordinate system~\eqref{etaeqn}. Firstly, we only need to keep track of a single variable $\eta_2(t)=u(t)+\mathbbm{i}\,v(t)\equiv\left(u(t),v(t)\right)$ associated with the vortex $\mathcal{V}_2$. Secondly, it is  easy to characterize self-similar evolutions as they correspond to the equilibrium points there (see, lemma~\ref{lemma_self_similarity}). 
Note that a vortex evolution is called self-similar evolution if the vortex coordinates satisfy $z_\alpha(t)=\lambda_\alpha f(t)$ with 
$f$ and $\lambda_\alpha$'s being 
an arbitrary complex function and complex constants, respectively~\cite{NS1979,K1987,TT1988,LK2000,H2007,H2008,
AH2010,G2016,KS2018}.
We will analyze various vortex trajectories by looking at the relative motion of vortex $\mathcal{V}_2$.
In order to do so, we write expressions of squared inter-vortex distances in terms of $u$ and $v$ using cosine rule~\eqref{cosine_rule} and definition~\eqref{etaeqn} \begin{equation}
\label{r2byr1_r12byr1}
r_2^2=r_1^2\,\left(u^2+v^2\right)\quad\mbox{and}\quad
r_{12}^2=r_1^2\left((u-1)^2+v^2\right).
\end{equation}
From~\eqref{eq_R2},~\eqref{eq_R1} and~\eqref{r2byr1_r12byr1}, we can write the differential equations for inter-vortex distances as
%
\begin{align}
\label{r1_sqr_dot_uv_eqn}
\dot{r_1^2}&=\frac{\Gamma_2v}{\pi\left((u-1)^2+v^2\right)},\\
\label{r2_sqr_dot_uv_eqn}
\dot{r_2^2}&=\frac{-\Gamma_1v}{\pi\left((u-1)^2+v^2\right)}.
\end{align}

 Note that if $v(t)\neq 0$ then $r_1(t)$ (similarly $r_2$) is either strictly increasing or decreasing at time $t$, with extrema existing only when a $(u,v)$ trajectory intersects the $u$-axis. Consequently, a trajectory in $(u,v)$ phase plane, which is bounded away from the $u$-axis, corresponds to either an unbounded vortex motion, or a vortex collapse situation.
Before we move on to characterize the entire $(u,v)$ phase plane trajectories, let us first look at the simplest one of them all; the equilibrium points. 

\begin{lemma}
\label{lemma_self_similarity}
 The constrained three-vortex system~\eqref{z0_z1_z2_dot} evolves self-similarly if and only if the corresponding trajectory in the $(u,v)$ phase plane is an equilibrium  solution, i.e., $\left(u(t),v(t)\right)=(u,v)|_{t=0}$.
\begin{proof}
If the vortex system evolves in a self-similar way, then recall that, there exists a complex valued function $f$ and complex constants $\lambda_\alpha$ ($\alpha=0,1,2$) such that $z_\alpha(t)=\lambda_\alpha f(t)$. 
WLOG, one may assume $f(0)=1$. Hence from the assumptions about the initial conditions we have $\lambda_0=0$, $\lambda_1=1$, and 
$\lambda_2\neq0$ (see,~figure~\ref{figure_model}(a)). Since $z_1(t)\neq 0$ for $t\in[0,t^*)$ and $t^*>0$, $f(t)\neq 0$ as long as the three-vortex problem  is defined. Therefore, $\eta_2=z_2/z_1=\lambda_2$, which is an equilibrium solution in the $(u,v)$ phase plane. The same lines of arguments, if retraced back, give the proof for the converse part.
 \end{proof}
\end{lemma}


Since the ratio of the inter-vortex distances remains constant in a self-similar evolution, the shape of the vortex triangle at any instant of time remains the same. A vortex trajectory in which both the size and shape of the vortices remain intact, and the vortex system move as a rigid body, is called a fixed configuration (also known as a vortex equilibrium). These solutions, therefore, must satisfy $\dot{r_1}=\dot{r_2}=\dot{r_{12}}=0$. 
Note that a fixed configuration may translate or rotate and since one of the vortices is fixed in the plane we cannot have translating fixed configurations.
It follows that fixed configurations also correspond to equilibrium points in the $(u,v)$ phase plane, but unlike the self-similar solutions they lie solely on the $u$-axis, as shown next. 
 
\begin{lemma}
\label{lemma_fixed_configuration}
The constrained three-vortex system~\eqref{z0_z1_z2_dot} is in a fixed configuration if and only if the corresponding trajectory in the $(u,v)$ phase plane is an equilibrium solution on the $u$-axis.
\begin{proof}
From~\eqref{eq_R2}--\eqref{eq_R01}, we see that any fixed configurations must have the area $A$
of the vortex triangle to be zero, i.e., it must be a collinear configuration. 
For a collinear configuration to remain collinear for all time, we also require  $\dot{A}=0$. Hence, $A=0$ and $\dot{A}=0$ are the necessary and sufficient conditions that the vortex system must satisfy in order 
to be in a fixed configuration. 
These two conditions in $(u,v)$ phase plane are written as 
$A=1/2\, r_1^2v=0$ 
and $\dot{A}=r_1\dot{r_1}v+ 1/2\, \dot{v} r_1^2=1/2\, \dot{v}r_1^2=0$, and   
therefore, $v=0$ and $\dot{v}=0$ are the corresponding necessary and sufficient conditions for fixed configuration in the $(u,v)$ phase plane. 
It turns out that $v=0$ implies $\dot{u}=0$, and hence the proof [see~\eqref{Meq0_uv} and~\eqref{Mneq0_udot_vdot}]. 
\end{proof}
\end{lemma}

In the classical three-vortex problem where none of the vortices are fixed, fixed configurations are either collinear types or equilateral triangles. 
It is interesting to note that once there is a fixed vortex in the three-vortex system, all fixed configurations are collinear. Depending on the zero and non-zero values of the angular-impulse constant $M$, 
the constrained three-vortex problem
is divided into symmetric and asymmetric cases, respectively.

%

\subsection{Symmetric case (\texorpdfstring{$M=0$}{M=0})}
\label{subsec:Meq0}

It follows from~\eqref{cons_ang} that
$M=0$ takes place only if the circulations $\Gamma_1$ and $\Gamma_2$ of vortices $\mathcal{V}_1$ and  $\mathcal{V}_2$ have opposite signs,  i.e., $\Gamma_1\Gamma_2<0$. 
In addition, $M=0$ assumption gives the following relation
\begin{equation}\label{Meq0_eqn_of_circle}
|\eta_2|^2=u^2+v^2=
\frac{|z_2|^2}{|z_1|^2}=\left(\frac{r_2}{r_1}\right)^2=
\kappa^2,
\end{equation}
where $\kappa=\sqrt{-\Gamma_1/\Gamma_2}>0$ is a constant. 
Thus, we have
\begin{equation}
\eta_2=|\eta_2|e^{i\theta}=\kappa e^{i\theta}
\label{Meq0_uv}
\end{equation}
where $\theta=\theta_2-\theta_1$.
Consequently, all the trajectories in $(u,v)$ phase plane must be on the circle of radius $\kappa$ centered at the origin. 

It is now enough to look at the dynamics of $\theta$ in $(u,v)$ phase plane to understand the qualitative behaviour of these trajectories, and the corresponding 
vortex motion in the physical ($x,y$) plane.
Using~\eqref{cosine_rule} and~\eqref{Meq0_eqn_of_circle}, we have
\begin{equation}
r_2^2=\kappa^2\,r_1^2\quad \mbox{and}\quad r_{12}^2 =  r_1^2\,(1+\kappa^2-2\kappa\cos\theta).
\label{eqn:r2_r12}
\end{equation}
%
Substituting~\eqref{eqn:r2_r12}
into~\eqref{cons_energy}, and rearranging the resulting expression, yields 
\begin{equation}
\label{Meq0_eqn_logr1}
\Gamma\log r_1^2=\tilde{H}-\Gamma_1\Gamma_2\log\left(1+\kappa^2-2\kappa\cos\theta\right),
\end{equation}
where $\Gamma= \Gamma_0\Gamma_1+\Gamma_0\Gamma_2+\Gamma_1\Gamma_2$, and $\tilde{H}=-(4\pi H+\Gamma_0 \Gamma_2 \log \kappa^2)$ are finite constant that can be determined from the initial conditions.
If $\Gamma\neq0$,
\eqref{Meq0_eqn_logr1} simplifies to 
\begin{equation}
\label{Meq0_eqn_r1}
r_1^2=\frac{E_0}{(1+\kappa^2-2\kappa\cos\theta)^\gamma},\quad
\mbox{where}\quad \gamma=\Gamma_1\Gamma_2/\Gamma\quad \mbox{and}\quad E_0=e^{\tilde{H}/\Gamma}.
\end{equation}
Transforming Hamiltonian canonical equations~\eqref{eqn:canonical_EQN_Hamiltonian} into polar coordinates by substituting $x_\alpha=r_\alpha \cos\theta_\alpha$ and $y_\alpha = r_\alpha \sin\theta_\alpha$ for $\alpha=1,2$, and applying chain rule,
we get
\begin{align}
\dot{\theta_1}&= - \frac{1}{\Gamma_1r_1}\frac{\partial {H}}{\partial r_1}\quad\mbox{and}\quad
\dot{\theta_2}= - \frac{1}{\Gamma_2r_2}\frac{\partial {H}}{\partial r_2}.
\label{eqn:thetadot_1_2}
\end{align}
Evaluating ~\eqref{eqn:thetadot_1_2} and using~\eqref{cosine_rule}, we obtain
\begin{align}
2\pi \Gamma_1r_1\dot{\theta_1}&=\frac{\Gamma_0\Gamma_1}{r_1}+\frac{\Gamma_1\Gamma_2}{r_{12}^2}\left(r_1-r_2\cos\theta\right),\label{Meq0_theta1dot}\\
2\pi \Gamma_2r_2\dot{\theta_2}&=\frac{\Gamma_0\Gamma_2}{r_2}+\frac{\Gamma_1\Gamma_2}{r_{12}^2}\left(r_2-r_1\cos\theta\right)\label{Meq0_theta2dot}.
\end{align}

Adding $r_1$ times~\eqref{Meq0_theta1dot} and $r_2$ times~\eqref{Meq0_theta2dot}, and simplifying the resultant equation by using the fact that $M=\Gamma_1r_1^2+\Gamma_2r_2^2=0$, yields
\begin{equation}\label{Meq0_eqn_theta_dot}
\dot{\theta}=\frac{-\Gamma}{2\pi\Gamma_1r_1^2}.
\end{equation}
Equation~\eqref{Meq0_eqn_theta_dot} dictates that $\theta$ is a constant when $\Gamma_0\Gamma_1+\Gamma_0\Gamma_2+\Gamma_1\Gamma_2:=\Gamma=0$, and for $\Gamma\neq 0$, it is a strictly increasing or  decreasing function of time.
Recall that if $\Gamma \neq 0$,~
$r_1^2$ is given by \eqref{Meq0_eqn_r1}, which further reduces~\eqref{Meq0_eqn_theta_dot} to 
the following  evolution equation
\begin{equation}
\label{Meq0_eqn_theta_dot2}
\dot{\theta}=
\frac{-\Gamma(1+\kappa^2-2\kappa\cos\theta)^\gamma}{2\pi\Gamma_1E_0}.
\end{equation}
By utilizing~\eqref{Meq0_eqn_logr1} and~\eqref{Meq0_eqn_theta_dot}, we can classify 
all the trajectories in the $(u,v)$ phase plane into three different classes: (i) self-similar evolution, (ii) unbounded dipole motion, and (iii) bounded periodic motion, which are explained below.


\subsubsection{Self-similar evolutions
(\texorpdfstring{$\Gamma=0$}{Gamma=0})}

If $\Gamma=0$, then from \eqref{Meq0_eqn_theta_dot2} it follows that $\dot{\theta}\equiv 0$. Hence the angle $\theta$ between the two vortices $\mathcal{V}_1$ and $\mathcal{V}_2$ remains constant. This means that the vortex triangles 
at any two instances of time are similar
, which leads to the case of self-similar evolution. 
In this case, the trajectory in the $(u,v)$ phase plane is an equilibrium  solution given by $\eta_2(t)=u_0+i\,v_0=\kappa e^{i\,\theta_0}$, where $\theta_0=\theta|_{t=0}$. Since all equilibrium solutions in the $(u,v)$ phase plane correspond to the self-similar evolutions in the physical plane (see lemma~\ref{lemma_self_similarity}) when $\Gamma=0$; irrespective of the initial conditions, the motion of the vortex system becomes self-similar in nature. 

From~\eqref{r1_sqr_dot_uv_eqn}, we get $\dot{r_1^2}=C$, where $C=\Gamma_2v_0/\left(\pi\left((u_0-1)^2+v_0^2\right)\right)$ is a constant.  Integrating $\dot{r_1^2}=C$ with respect to time yields
\begin{equation}
\label{Meq0_self_similar_r1_r2_equation}
r_1(t)=\sqrt{1+Ct} \quad \text{and }\quad r_2(t)=\kappa\sqrt{1+Ct},\quad t>0.
\end{equation} 
Depending on the sign of $C$, determined by the the initial conditions, there are three possible scenarios 
%
\begin{enumerate}[label=(\roman*)]
\item
{Self-similar collapse ($C<0$):} 
The vortices $\mathcal{V}_1$, $\mathcal{V}_2$ move towards the fixed vortex $\mathcal{V}_0$, and precisely at time $t^*=-1/C>0$ they collide on it. This special kind of motion is called a self-similar collapse of the vortices. After the collision, the point vortex model breaks down, and no further analysis is possible.


\item
{Self-similar expansion $(C>0)$:} 
The vortices $\mathcal{V}_1$, $\mathcal{V}_2$ move further and further away from the fixed vortex, and hence the motion becomes unbounded as $t$ tends to infinity. 


\item
{Fixed collinear configuration ($C=0$):} 
The initial configuration is collinear, i.e., $v_0=0$. Since $\Gamma=0$ results in an  equilibrium solution in $(u,v)$ phase plane, we have $v(t)=v_0\equiv 0$, i.e., its an equilibrium solution on the $u$-axis. 
Hence from lemma~\ref{lemma_fixed_configuration}, this is a case of fixed collinear configuration.
In short, the vortices $\mathcal{V}_1$, $\mathcal{V}_2$  evolve in a circular fashion around the fixed vortex $\mathcal{V}_0$ with constant radii preserving the initial collinearity.
\end{enumerate}

\subsubsection{Unbounded dipole motion (\texorpdfstring{$\kappa=1$}{k=1}): Equal counter-rotating pair}
\label{subsebsec:Meq0_keq1}

If $\kappa=1$ then $\Gamma=\Gamma_1\Gamma_2\neq0$ and $\gamma=1$, which simplifies~\eqref{Meq0_eqn_theta_dot2} to $\dot{\theta}=\sigma (1-\cos\theta)$,
where $\sigma=4\Gamma_2/E_0$ is a non-zero constant. Integrating 
$\dot{\theta}$ 
and applying the initial condition $\theta_0=\theta|_{t=0}$, we arrive at 
\begin{equation*}
\cot(\theta_0/2)-\cot(\theta/2)=\sigma t.
\end{equation*} 
Hence, in the limit when $t$ tends to infinity, the angle $\theta$ tends to zero. From~\eqref{Meq0_eqn_r1} we see that $r_1$ becomes unbounded when $\theta\to 0$, and thereby leading to an unbounded motion for both the vortices $\mathcal{V}_1$ and $\mathcal{V}_2$
[note that for $\kappa=1$, $(u-1)^2+v^2=1+\kappa^2-2\kappa\cos\theta \to 0$ as $\theta \to 0$].
Therefore, a trajectory in the $(u,v)$ phase plane is a circular arc that asymptotically approaches the singularity point $(1,0)$.

Irrespective of the initial position of $\eta_2$ on the unit circle, the equal counter-rotating pair will always lead an unbounded motion. Although unbounded, in contrast to self-similar expansion, 
the distance between the free vortices does not increase with time. In fact $r_{12}$ remains a constant throughout the motion of vortices, which can be seen from~\eqref{cons_energy} by using the fact that $r_2=\kappa r_1=r_1$.
%
%
%

\subsubsection{Bounded periodic motions (\texorpdfstring{$\kappa \neq 1$, $\Gamma\neq0$}{k neq 1, gamma neq 0})}

If the right-hand side term of \eqref{Meq0_eqn_logr1}  is bounded, then $r_1$ is bounded. Note that the right-hand side term of~\eqref{Meq0_eqn_logr1} is unbounded only when $1+\kappa^2-2\kappa\cos\theta$ tends to zero. Since all the $(u,v)$ phase plane trajectories lie on a circle of radius $\kappa$, we have $|u|=|\kappa\cos\theta|\leq \kappa$, which implies
\begin{equation*}
1+\kappa^2-2u\geq 1+\kappa^2-2\kappa=(1-\kappa)^2.
\end{equation*}
Hence $1+\kappa^2-2u=0 \iff \kappa=1=u$. As we assume $\kappa\neq 1$, $1+\kappa^2-2u$ is never zero, and hence, $r_1$ is bounded on both sides. Moreover, as the sign of $\dot{\theta}$ remains unchanged from~\eqref{Meq0_eqn_theta_dot}, the trajectory in the $(u, v)$ phase plane must be a full circle. Furthermore, a closed trajectory implies periodicity in the $\theta$ variable, and therefore, periodicity in the inter-vortex distances $r_1$, $r_2$ and $r_{12}$.

\subsection{Examples for \texorpdfstring{$M=0$}{M=0} case}
\label{subsec:Examples_Meq0}

In this section, we shall illustrate 
different kinds of vortex trajectories as discussed in Sec.~\ref{subsec:Meq0}. To do so, we solve~\eqref{z0_z1_z2_dot} numerically using the fourth-order Runge-Kutta method for different initial conditions, and plot the obtained numerical solution in the $(u,v)$ as well as in the physical plane $(x,y)$.

\subsubsection{Self-similar evolutions (\texorpdfstring{$\Gamma=0$}{Gamma=0})}
\begin{figure}[htbp!]
\centering
\subfigure[]{
\includegraphics[width=0.32\textwidth,height=0.32\textwidth]{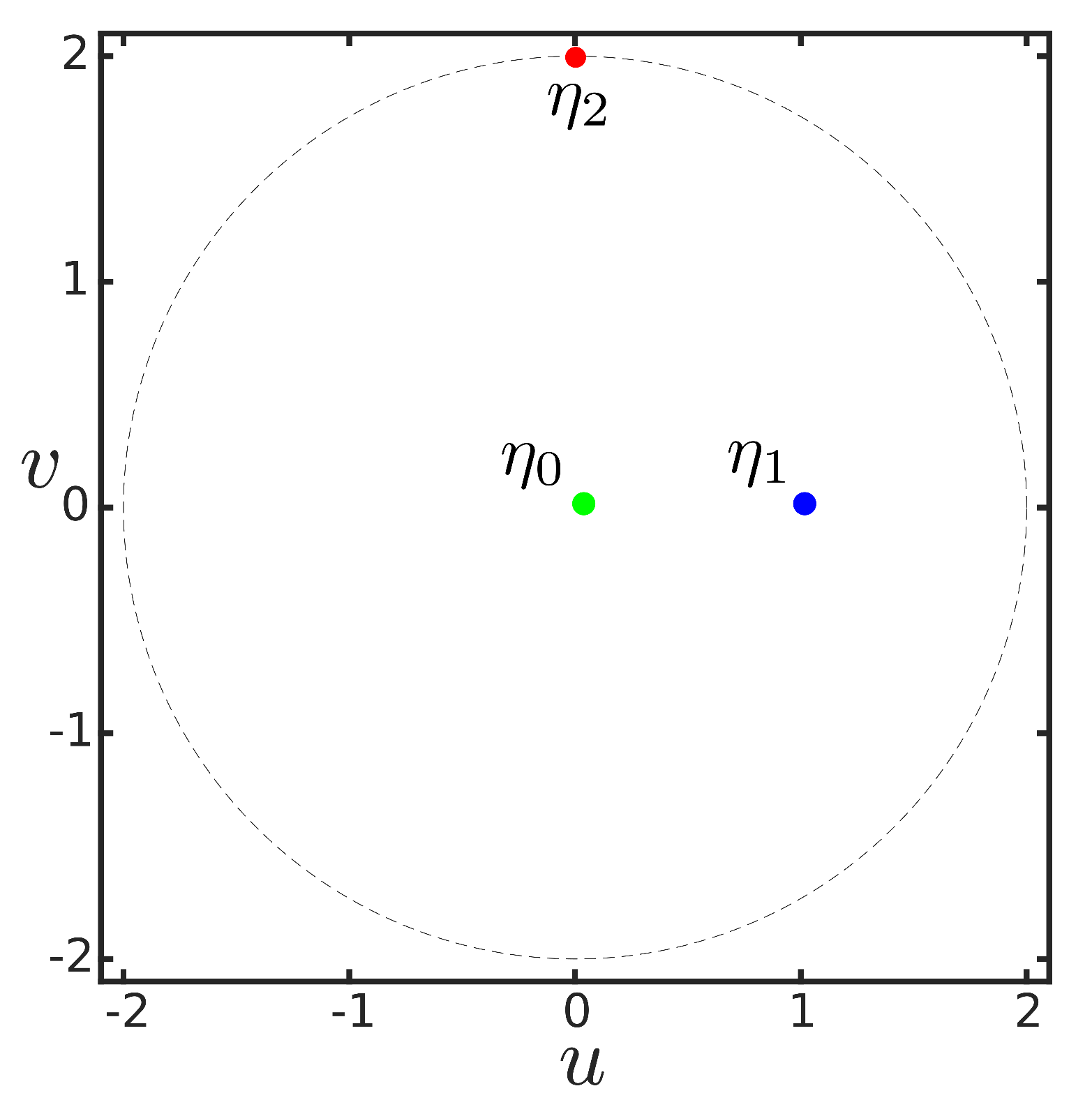}}
\subfigure[]{
\includegraphics[width=0.32\textwidth,height=0.32\textwidth]{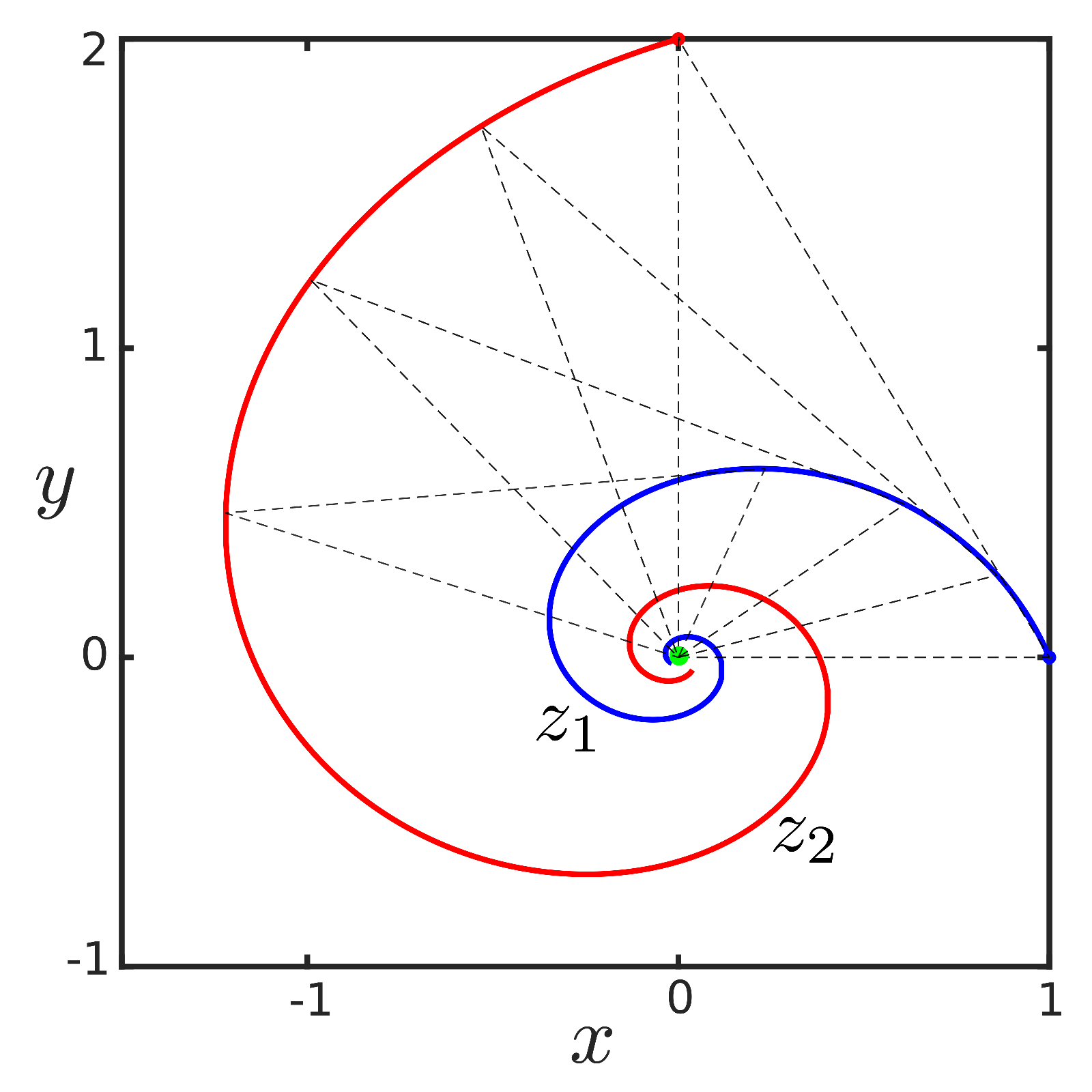}}
\subfigure[]{
\includegraphics[width=0.32\textwidth,height=0.32\textwidth]{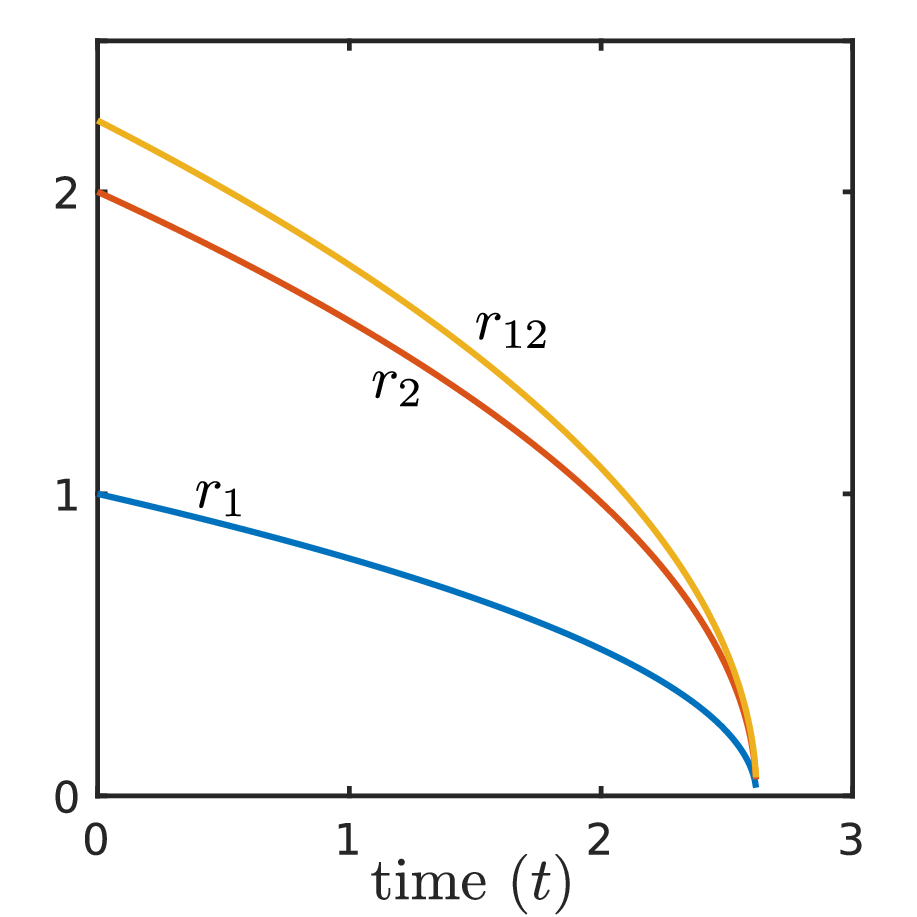}}
\caption{\small{Phase diagram showing a trajectory in the case of self-similar collapse. The positions of vortices $\mathcal{V}_0$, $\mathcal{V}_1$ and $\mathcal{V}_2$ are marked by green, blue, and red, respectively in (a) $(u,v)$ phase plane and (b) $(x,y)$ physical plane. (c) The inter-vortex distance functions are plotted against time.}}
\label{fig_Meq0_vortex_collapse}
\end{figure} 
We consider the vortex circulations $(\Gamma_0,\Gamma_1,\Gamma_2)=(4,12,-3)$ which satisfy the equality $\Gamma=\Gamma_0\Gamma_1+\Gamma_0\Gamma_2+\Gamma_1\Gamma_2=0$. 
Using conditions $z_1|_{t=0} =(1,0)$ and 
$M=\Gamma_1r_1^2+\Gamma_2r_2^2=0$, we get $r_2|_{t=0}=2$. This implies that 
we need to choose $z_2|_{t=0}$ from the circle of radius two centered at the origin. 
Note that because of the choice of $z_1|_{t=0}=(1,0)$, we also have $z_2|_{t=0}=\eta_2|_{t=0}$.

\begin{enumerate}[label=(\roman*)]
\item
{Self-similar collapse ($C<0$):}
For the initial conditions $z_1|_{t=0}=(1,0)$ and  $z_2|_{t=0}=(0,2)$ we get $C=-6/5\pi<0$, which correspond 
a self-similar collapse. Hence, we would expect both $r_1$ and $r_2$ to decrease monotonically to zero, and  at $t^*=5\pi/6\approx 2.618$, the free vortices $\mathcal{V}_1$ and $\mathcal{V}_2$  to collide with the fixed vortex $\mathcal{V}_0$. Plotting  the numerical solution for $0\leq t\leq 2.616$ yields figure~\ref{fig_Meq0_vortex_collapse}. The $(u,v)$ phase plane trajectory [see figure~\ref{fig_Meq0_vortex_collapse}(a)] is a single point (marked red), indicating that it is an equilibrium solution. In the physical plane [see figure~\ref{fig_Meq0_vortex_collapse}(b)], we see that the vortices move towards the fixed vortex in a spiral fashion. The vortex triangle formed by joining the vortices $\mathcal{V}_0$, $\mathcal{V}_1$, and $\mathcal{V}_2$ is shown by dashed lines at four different time 
As expected, they are all similar triangles with decreasing area. In figure~\ref{fig_Meq0_vortex_collapse}(c), the inter-vortex distance functions $r_1$, $r_2$, and $r_{12}$ are seen monotonically decreasing and simultaneously reaching the zero value in finite time, agreeing with our analysis presented in Sec.~\ref{subsec:Meq0}.

\item
\begin{figure}
\subfigure[]{
\includegraphics[width=0.32\textwidth,height=0.32\textwidth]{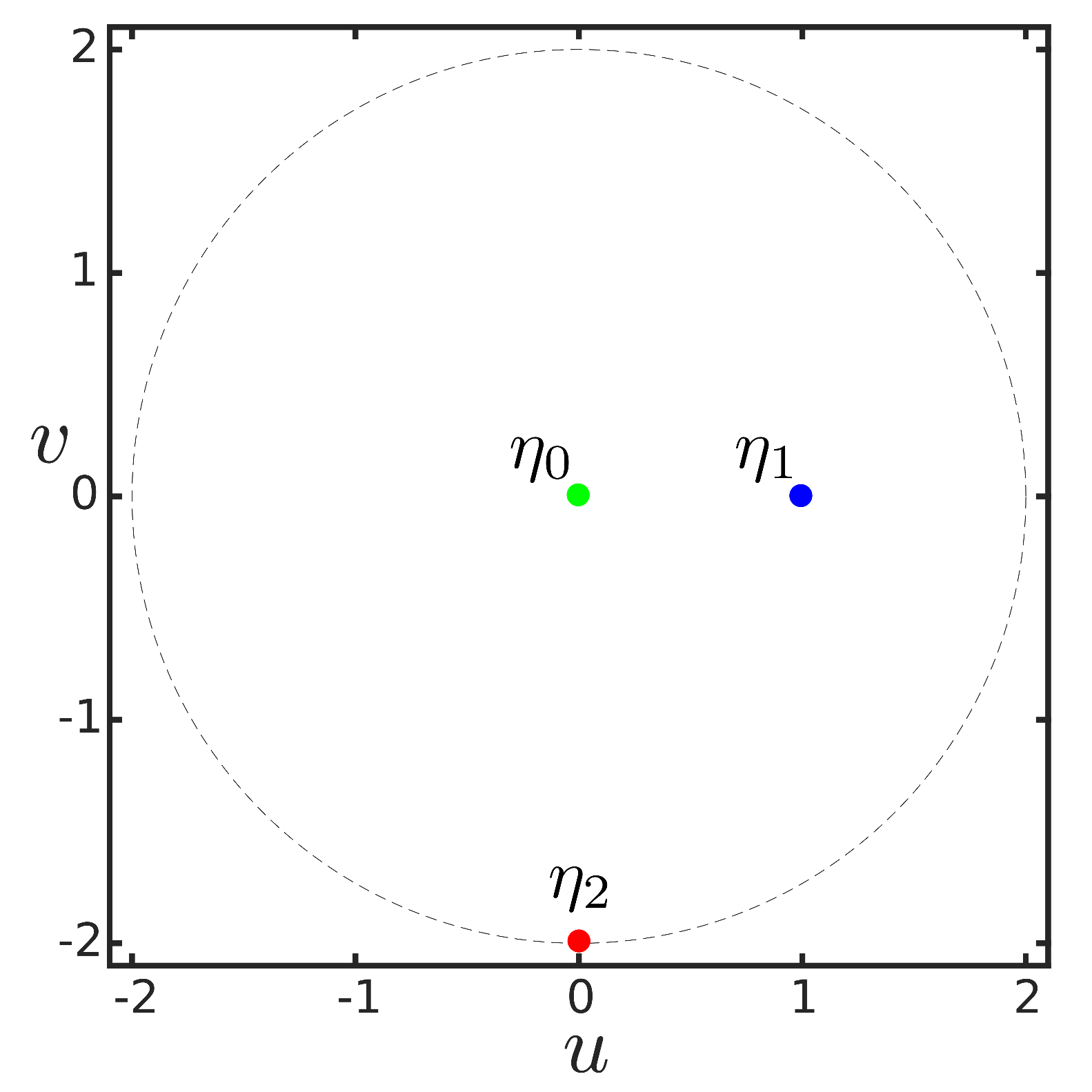}}
\subfigure[]{
\includegraphics[width=0.32\textwidth,height=0.32\textwidth]{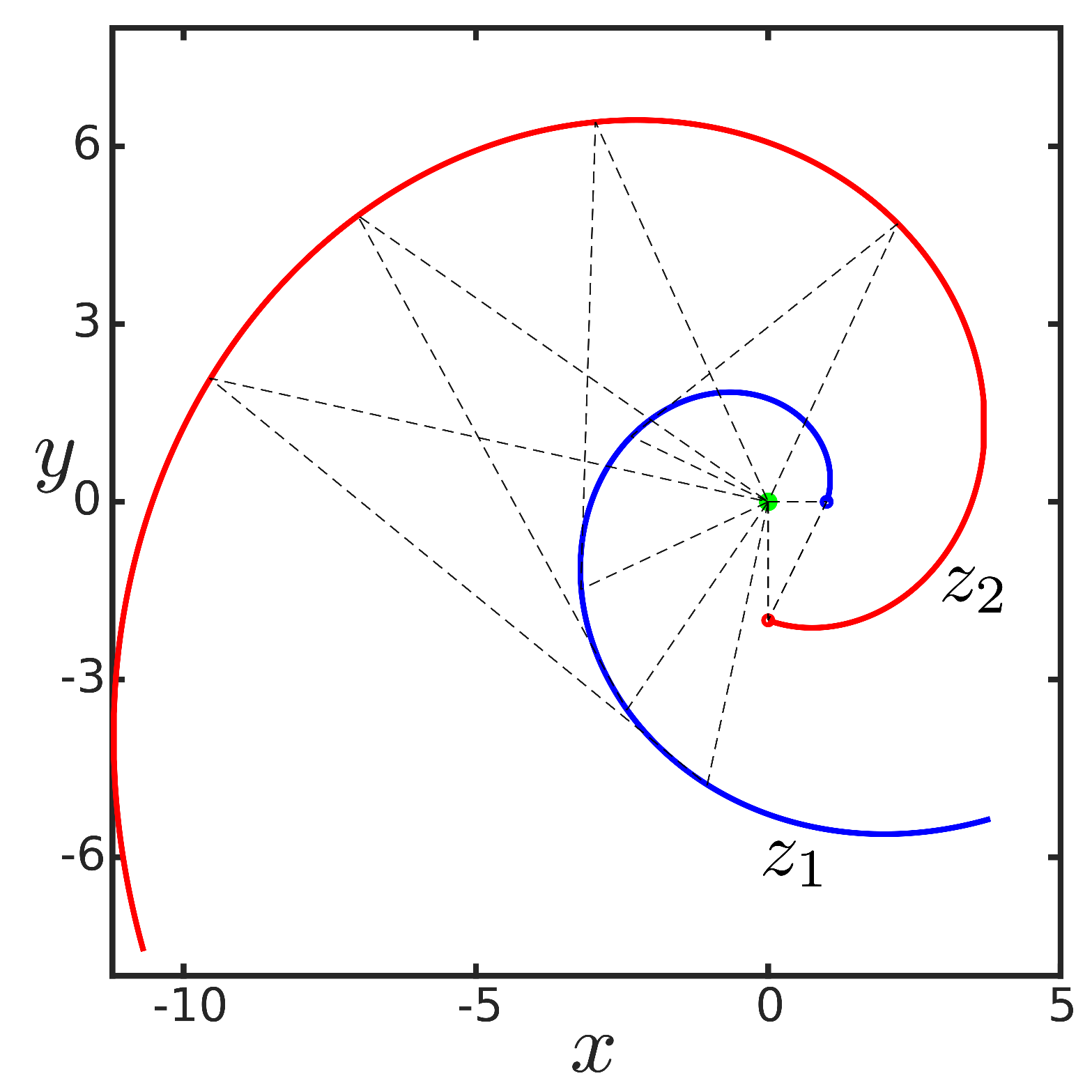}}
\subfigure[]{
\includegraphics[width=0.32\textwidth,height=0.32\textwidth]{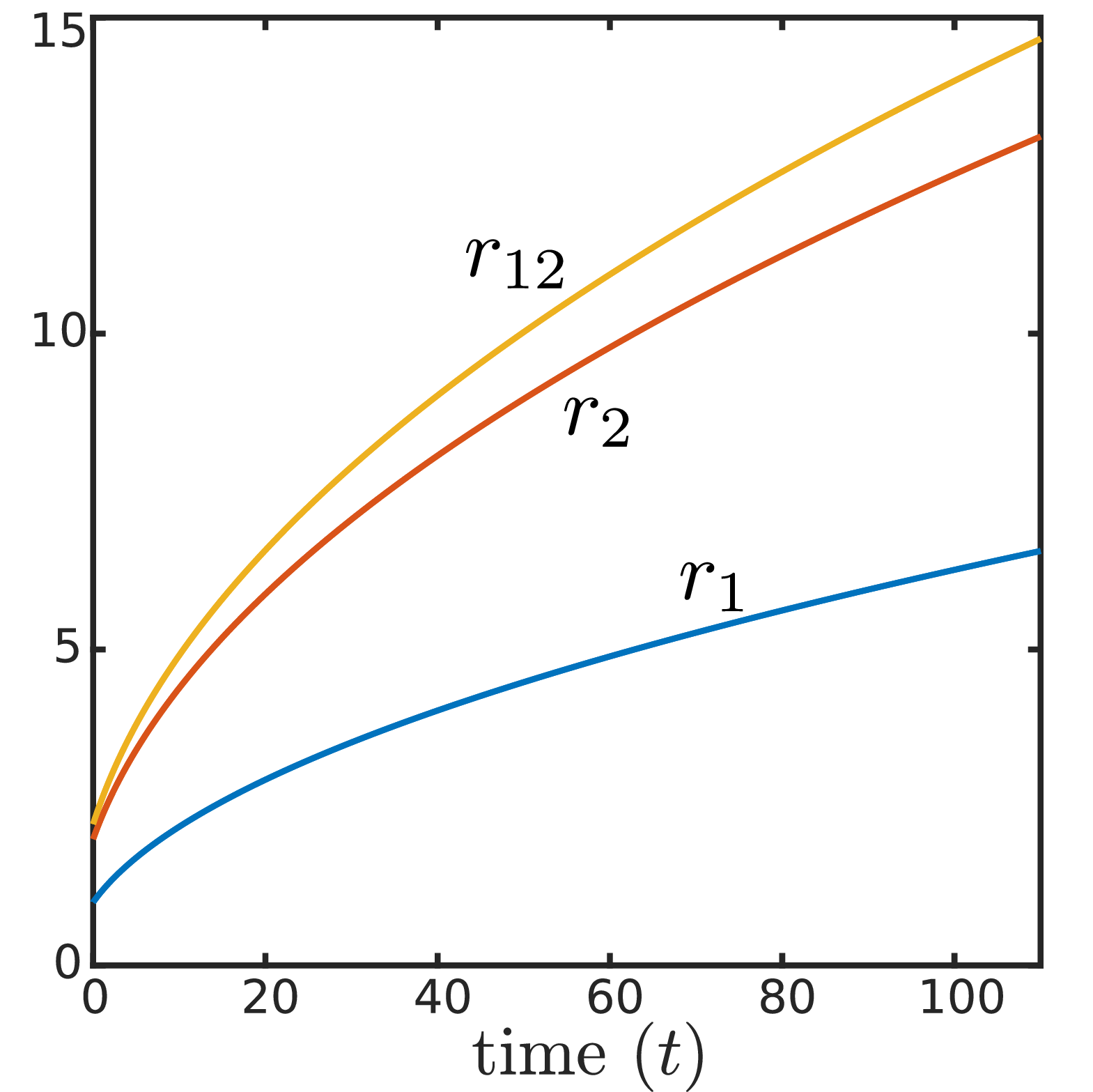}}
\caption{\small{
Same as figure~\ref{fig_Meq0_vortex_collapse} but for the case of self-similar expansion.
}}
\label{fig_Meq0_vortex_expansion}
\end{figure}

{Self-similar expansion ($C>0$):} For the initial conditions $z_1|_{t=0}=(1,0)$ and $ z_2|_{t=0}=(0,-2)$, we get $C=6/5\pi>0$, which correspond to a self-similar expansion. Thus, one would expect the free vortices $\mathcal{V}_1$ and $\mathcal{V}_2$ to move away from the fixed vortex, maintaining the angle between them for $t>0$. 
Plotting the numerical solution for $t\in [0,110]$  yields figure~\ref{fig_Meq0_vortex_expansion}. The $(u,v)$ phase plane trajectory is an equilibrium solution at $(0,-2)$ [see figure~\ref{fig_Meq0_vortex_expansion}(a)] whereas the actual vortex trajectory consists of free vortices $\mathcal{V}_1$ and $\mathcal{V}_2$ moving away from each other in a spiral fashion around the fixed vortex $\mathcal{V}_0$ [see figure~\ref{fig_Meq0_vortex_expansion}(b)]. Four dashed triangles in figure~\ref{fig_Meq0_vortex_expansion}(b), formed by joining the vortex positions at four different 
time, show that the vortex triangles remain similar but with increasing area. In figure~\ref{fig_Meq0_vortex_expansion}(c), it is seen that the inter-vortex distance functions $r_1,r_2$, and $r_{12}$ are monotonically increasing with time.

\item
{Fixed collinear configurations ($C=0$):}
For the initial conditions $z_1|_{t=0}=(1,0)$ and $z_2|_{t=0}=(2,0)$, we get $C=0$. The inter-vortex distances $r_1,r_2$, and $r_{12}$ remain constant [see figure~\ref{fig_Meq0_fixed_conf}(c)] throughout the motion indicating a fixed configuration. Note that all three vortices lie on the $x$-axis, and hence, are collinear initially. The numerical solution for the given initial conditions shows that the $(u,v)$ phase plane trajectory is an equilibrium point on the $u$-axis at $(2,0)$ [see figure~\ref{fig_Meq0_fixed_conf}(a)]. In the physical plane, the free vortices $\mathcal{V}_1$ and $\mathcal{V}_2$ move in circular orbits around the fixed vortex $\mathcal{V}_0$, as shown in figure~\ref{fig_Meq0_fixed_conf}(b). The vortex positions at several instances are joined by the dashed lines [see figure~\ref{fig_Meq0_fixed_conf}(b)], and it is evident that the vortices retain the collinearity throughout the motion. 
\begin{figure}[!b]
\subfigure[]{
\includegraphics[width=0.32\textwidth,height=0.32\textwidth]{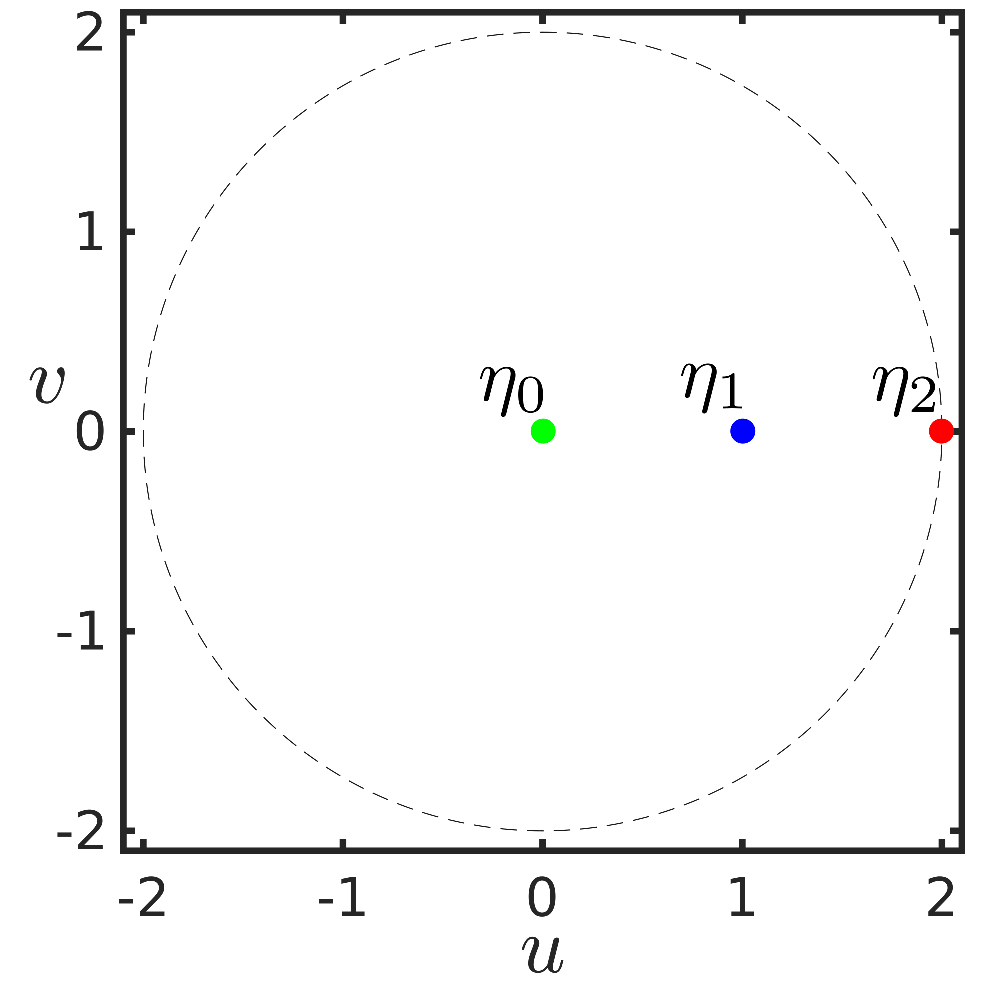}}
\subfigure[]{
\includegraphics[width=0.32\textwidth,height=0.32\textwidth]{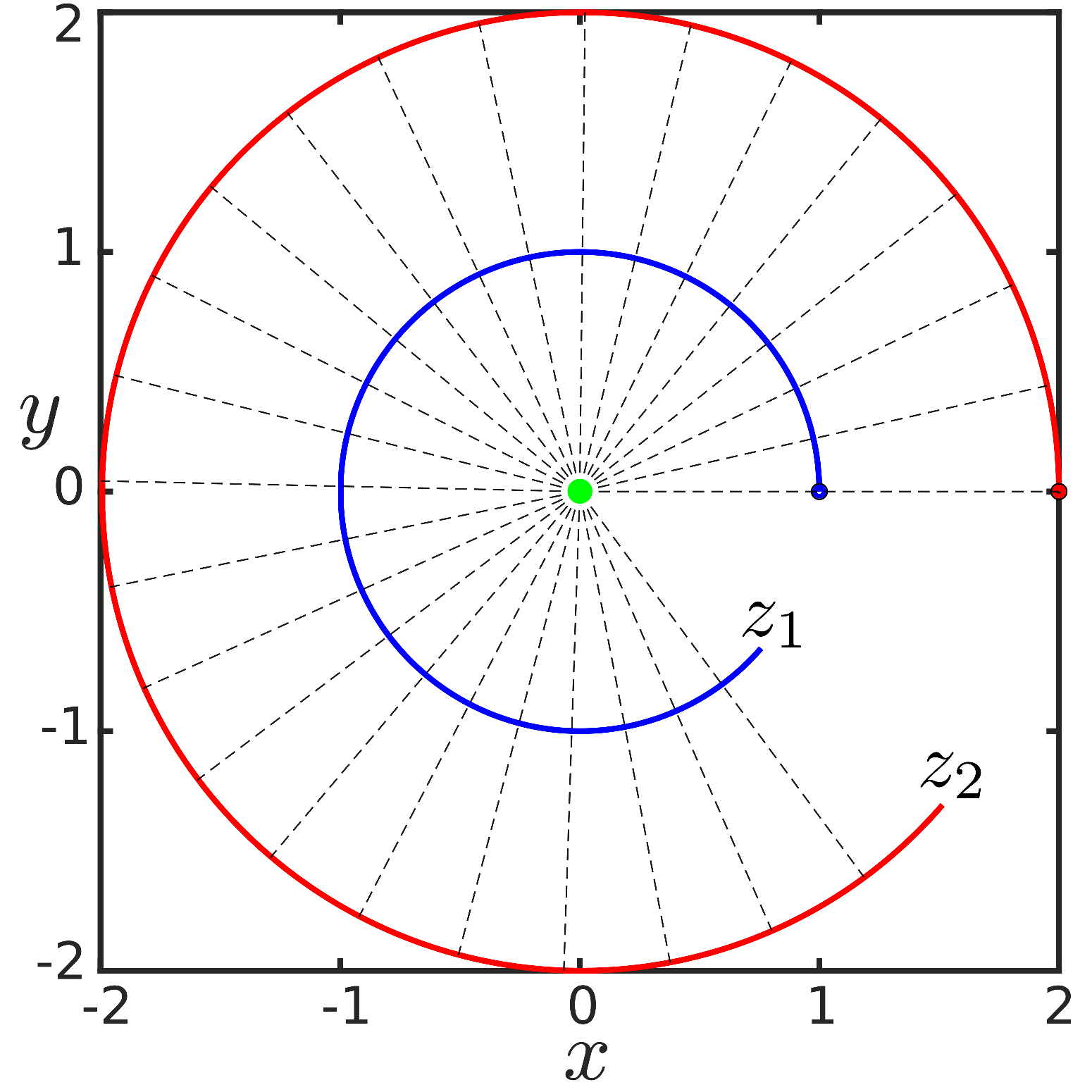}}
\subfigure[]{
\includegraphics[width=0.32\textwidth,height=0.32\textwidth]{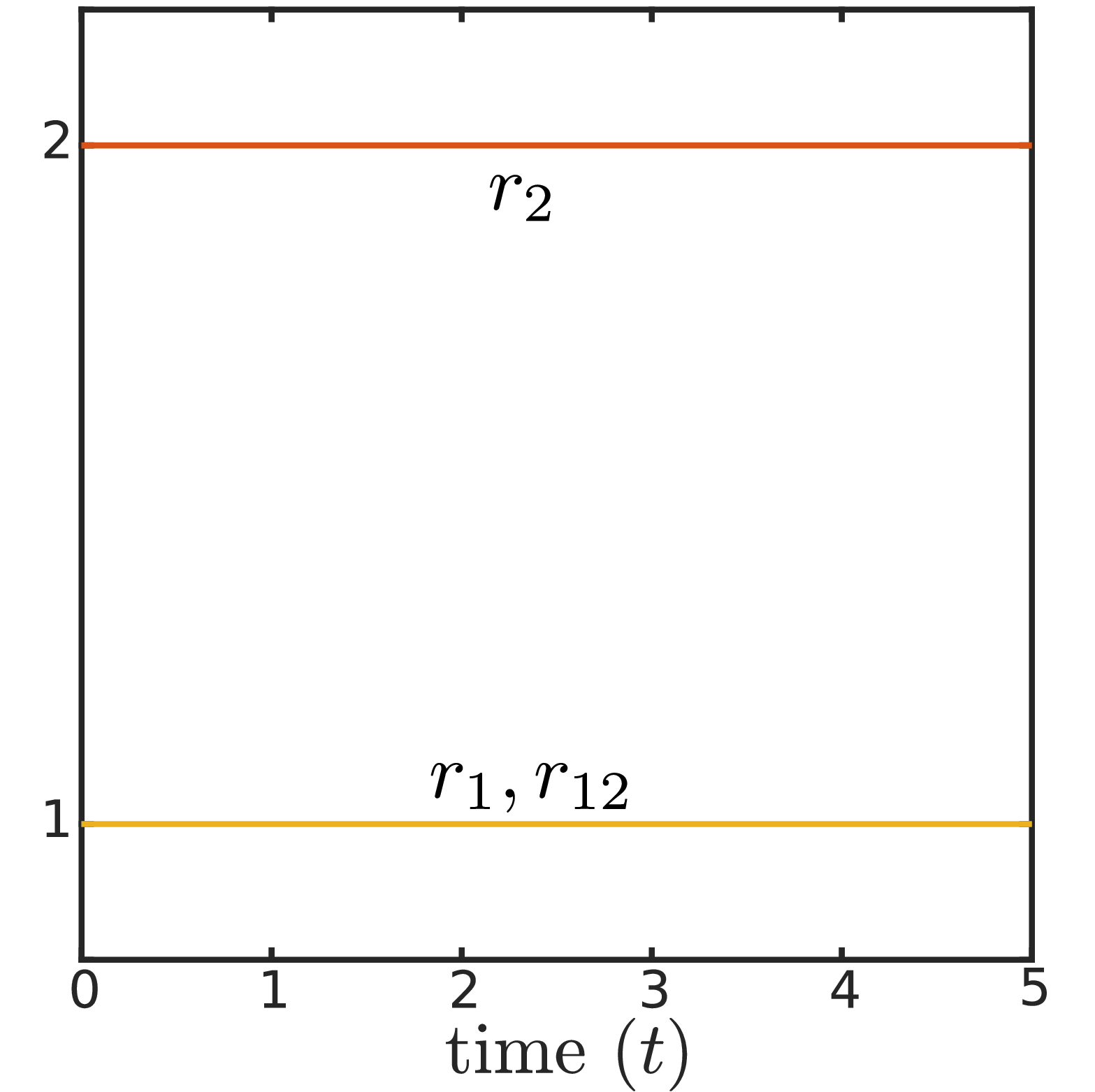}}
\caption{\small{
Same as figure~\ref{fig_Meq0_vortex_collapse} but for the case of fixed configuration.
}}
\label{fig_Meq0_fixed_conf}
\end{figure}
\end{enumerate}

\subsubsection{Unbounded dipole motion (\texorpdfstring{$\kappa=1$}{k=1}): Equal counter-rotating pair}

Let us illustrate the vortex motion when the circulations $(\Gamma_0,\Gamma_1,\Gamma_2)$ are given by $(2,1,-1)$ such that $\kappa=\sqrt{-\Gamma_1/\Gamma_2}=1$. Since $z_1|_{t=0}=(1,0)$, the initial condition for $z_2$ must be on a circle of radius $1$ centered at the origin, so as to make $M=0$. Figure~\ref{fig_Meq0_counter_rotating} shows the numerical solution of~\eqref{z0_z1_z2_dot} for the given set of parameters and $z_2|_{t=0}=(0,1)$. 
The $(u,v)$ phase plane trajectory [see figure~\ref{fig_Meq0_counter_rotating}(a)] remains on a circle of radius $1$ (dashed line), and moves towards the singularity point $\eta_1$ asymptotically. 
From~\eqref{r1_sqr_dot_uv_eqn}--\eqref{r2_sqr_dot_uv_eqn}, we observe that in finite time the trajectory intersects 
$u$-axis at $(-1,0)$, which corresponds to a minimum for the inter-vortex distance functions.
Thus, after achieving the minimum, $r_1$ and $r_2$ must monotonically increase with time. 
This is justified in figure~\ref{fig_Meq0_counter_rotating}(c), which illustrates the variation of inter-vortex distances with time. As mentioned before in Sec.~\ref{subsebsec:Meq0_keq1}, $r_{12}$ is a constant function, and $r_1=r_2$. In the $(x,y)$ plane [see figure~\ref{fig_Meq0_counter_rotating}(b)], the free vortices are seen moving away from the fixed vortex, indicating the unbounded nature of the vortex motion. 
\begin{figure}[htb!]
\centering
\subfigure[]{
\includegraphics[width=0.32\textwidth,height=0.32\textwidth]{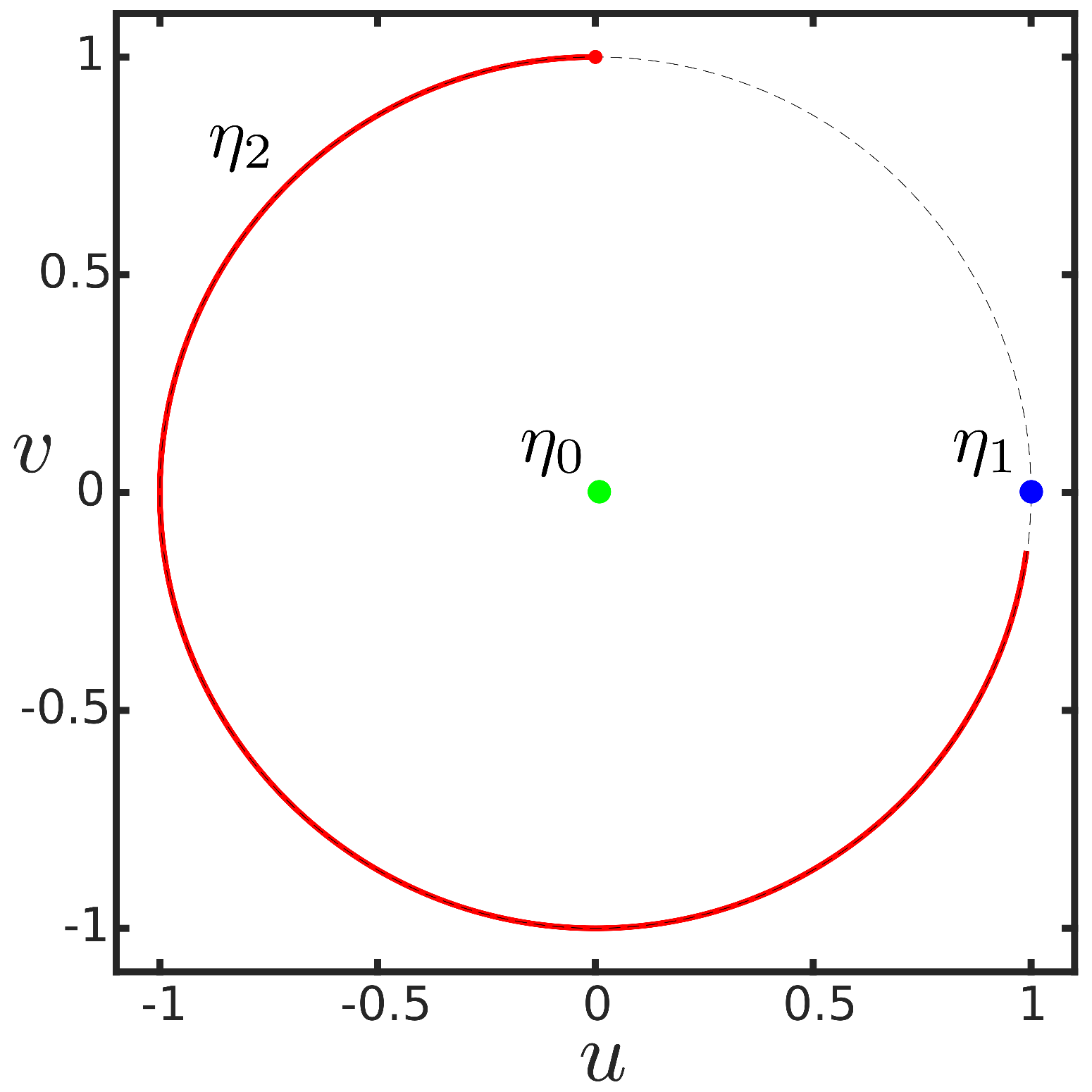}}
\subfigure[]{
\includegraphics[width=0.32\textwidth,height=0.32\textwidth]{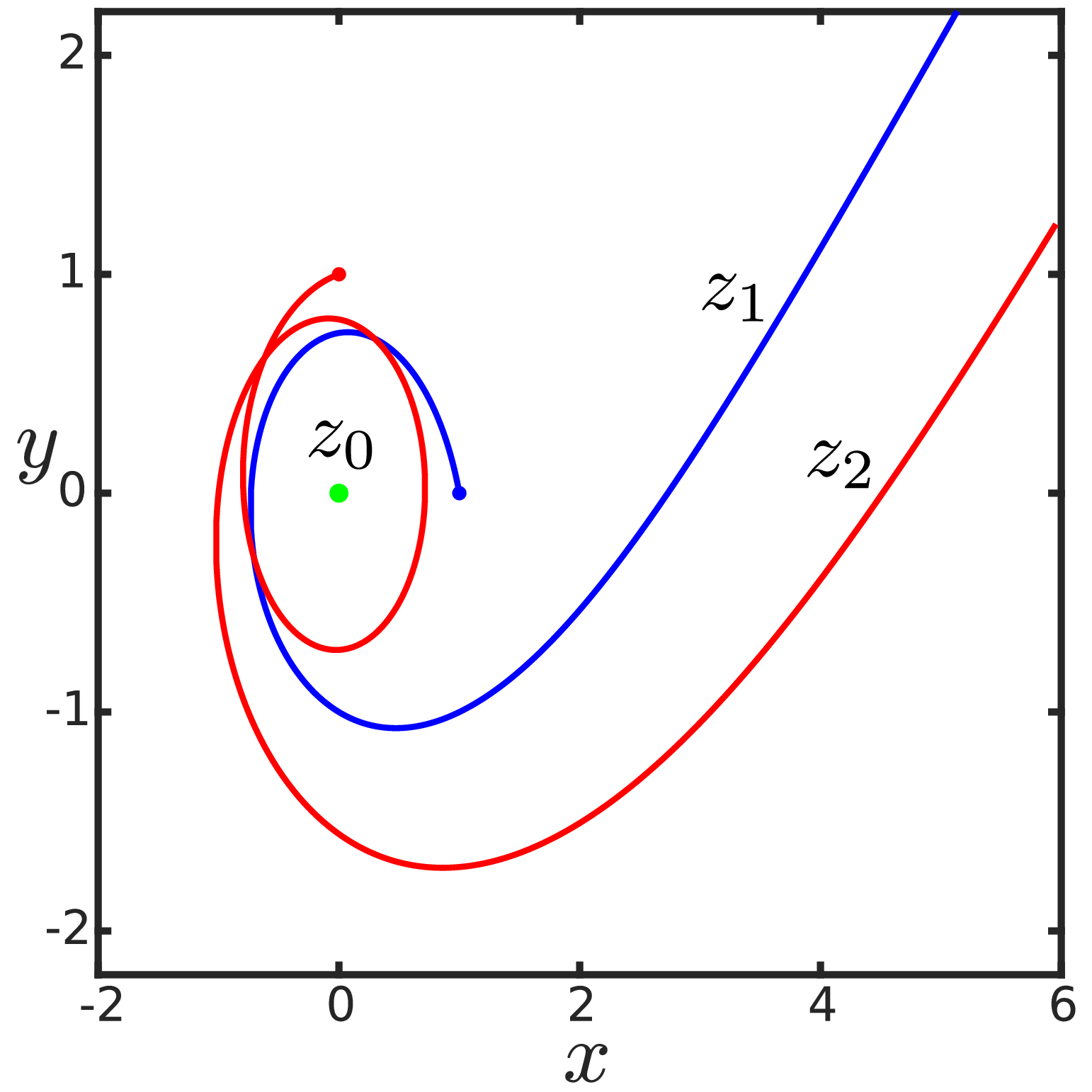}}
\subfigure[]{
\includegraphics[width=0.32\textwidth,height=0.32\textwidth]{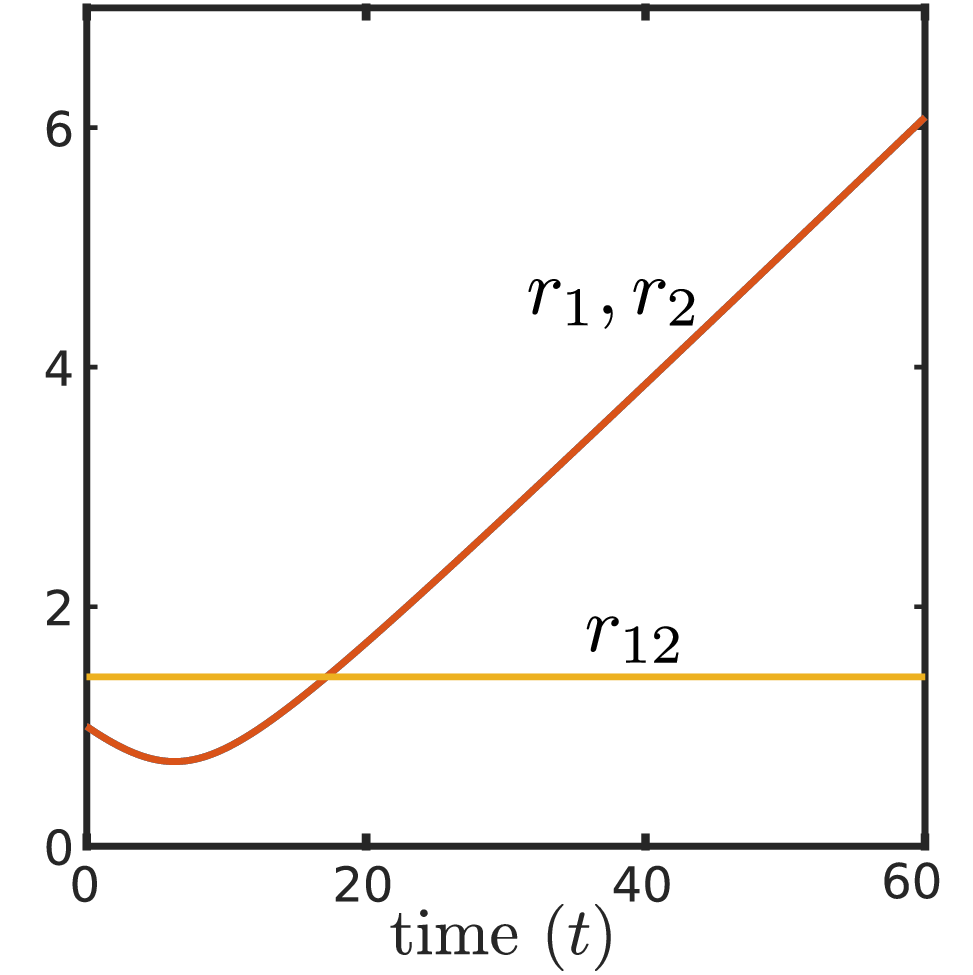}}
\caption{\small{
Same as figure~\ref{fig_Meq0_vortex_collapse} but for the case of equal counter-rotating pair of free vortices.
}}
\label{fig_Meq0_counter_rotating}
\end{figure}

\subsubsection{Bounded periodic motion \texorpdfstring{($\kappa\neq 1$, $\Gamma\neq0$)}{(kappa neq 1, Gamma neq 0}}

Let us now analyze the vortex motion when the circulations are $(\Gamma_0,\Gamma_1,\Gamma_2)=(-3,12,-2.9)$ for the initial conditions $z_1|_{t=0}=(1,0)$ and $z_2|_{t=0}=(\sqrt{12/2.9},0)$.
Clearly 
$M=0$, $\Gamma\neq 0$ and $\kappa\neq 1$ in this case.
The $(u,v)$ phase plane trajectory is a  full circular orbit of radius $\sqrt{12/2.9}$ centered at the origin [see figure~\ref{fig_Meq0_periodic}(a)], whereas the actual vortex motion consists of free vortices moving around the fixed vortex $\mathcal{V}_0$ with repeating patterns, indicating the periodicity in the inter-vortex distance functions [see figure~\ref{fig_Meq0_fixed_conf}(b)]. Periodicity in the variables $r_1,r_2$, and $r_{12}$ is also confirmed in figure~\ref{fig_Meq0_periodic}(c).
\begin{figure}[htbp!]
\centering
\subfigure[]{
\includegraphics[width=0.32\textwidth,height=0.32\textwidth]{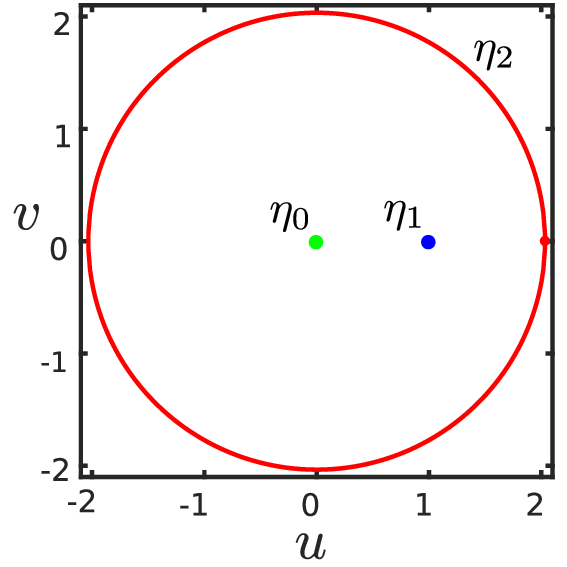}}
\subfigure[]{
\includegraphics[width=0.32\textwidth,height=0.32\textwidth]{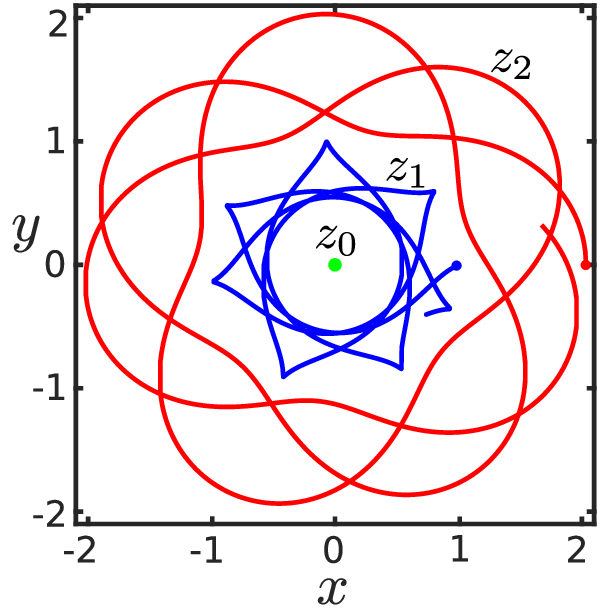}}
\subfigure[]{
\includegraphics[width=0.32\textwidth,height=0.32\textwidth]{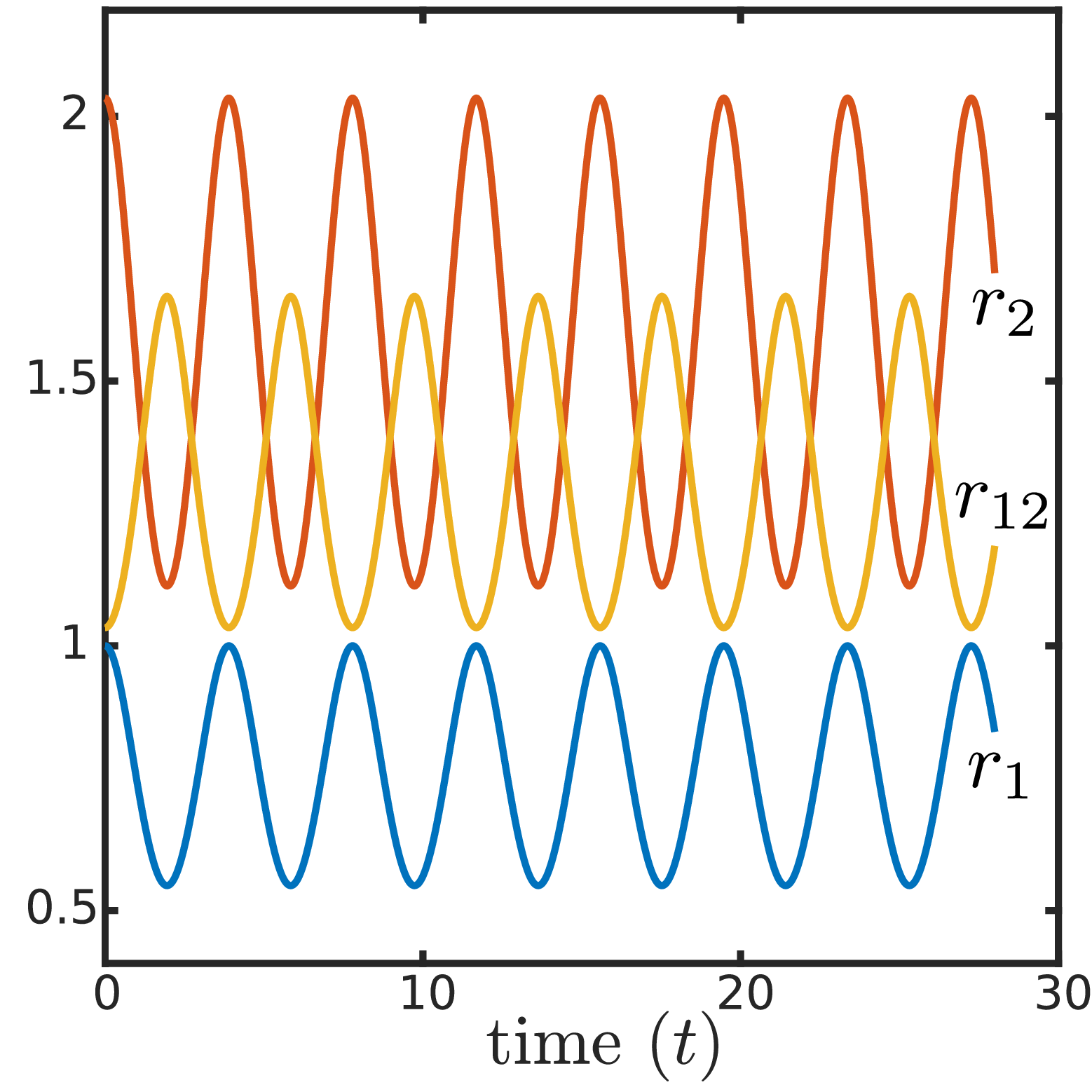}}
\caption{\small{
Same as figure~\ref{fig_Meq0_vortex_collapse} but for the case of periodic inter-vortex distances.}}
\label{fig_Meq0_periodic}
\end{figure}

\subsection{Asymmetric case (\texorpdfstring{$M\neq0$}{M neq 0})}
\label{subsec:Mneq0}

 Let us begin this section by observing that vortex collisions never happen in the case of non-zero $M$. A finite value of the Hamiltonian~\eqref{cons_energy} along with the inequality $|r_1-r_2|\leq r_{12}\leq r_1+r_2$ implies that
$r_1$ tends to zero if and only if $r_2$ tends to zero. This indeed cannot happen, as $M=\Gamma_1r_1^2+\Gamma_2r_2^2$ is assumed to be non-zero. Therefore, the inter-vortex distance functions $r_1$ and $r_2$ are bounded away from zero.

\begin{lemma}\label{lemma_Mneq0_lowerbound_r1r2}
 For the asymmetric $M\neq 0$ case, the  inter-vortex distances $r_1$ and $r_2$ are bounded away from zero for all time.
  \begin{proof}We shall give a proof by contradiction. Suppose there exists a real sequence $\{t_n\}_{n\in\mathbb{N}}$ such that $\lim_{n\to \infty}r_1(t_n)=0 $. From~\eqref{cons_ang} it follows that $\lim_{n\to \infty}r_2(t_n)=\sqrt{M/\Gamma_2}$. Since the inequality $|r_1-r_2|\leq r_{12}\leq r_1+r_2$ must hold for all time, the sequence $\{r_{12}(t_n)\}_{n\in\mathbb{N}}$ must also tend to $\sqrt{M/\Gamma_2}$. The contradiction is that the left-hand side of~\eqref{cons_energy} is finite but the right-hand side is not. Similarly, one may argue for the case $r_2\to 0$ to arrive at a contradiction. In fact, $r_{12}$, the distance between the vortices $\mathcal{V}_1$ and $\mathcal{V}_2$ is also bounded away from zero (see lemma~\ref{lemma_Mneq0_r12_bounded}).
 \end{proof}
\end{lemma}

Next, we will derive the underlying differential equations in the $(u,v)$ phase plane. Differentiating $\eta_2(t)=z_2(t)/z_1(t)$ with respect to time, we get
\begin{align}
\dot{\eta_2}&=\frac{\dot{z_2}}{z_1}+\frac{z_2\dot{\bar{z}}_1}{r_1^2}+z_2  \bar{z}_1 \dot{\left(\frac{1}{r_1^2}\right)},
\end{align}
where bar over a variable denotes its complex conjugate. 
Substituting $\dot{z_1}$ and $\dot{z_2}$ from~\eqref{z0_z1_z2_dot} into above equation and simplifying, we get the following expression
\begin{align}
\dot{\eta_2}&=\frac{\mathbbm{i}\Gamma_0}{2\pi}\eta_2\left(r_2^{-2}-r_1^{-2}\right)+\frac{\mathbbm{i}(\Gamma_1-\Gamma_2)}{2\pi}\frac{\eta_2}{r_{12}^2}+\frac{\mathbbm{i}\Gamma_2}{2\pi}\frac{|\eta_2|^2}{r_{12}^2}+\eta_2r_1^2\dot{\left(\frac{1}{r_1^2}\right)}.
\label{Mneq0_eta2dot}
\end{align}

Since, $M\neq 0$, using~\eqref{cons_ang} and~\eqref{r2byr1_r12byr1} it is possible to express the inter-vortex distances $r_1,r_2$, and $r_{12}$ in terms of coordinates $u$ and $v$ as
\begin{align}
r_1^2&=\frac{M}{\Gamma_1+\Gamma_2(u^2+v^2)},\quad
r_2^2=\frac{M(u^2+v^2)}{\Gamma_1+\Gamma_2(u^2+v^2)},\quad 
r_{12}^2=\frac{M((u-1)^2+v^2)}{\Gamma_1+\Gamma_2(u^2+v^2)}.
\label{Mneq0_r1_r2_r12}
\end{align}
Substituting~\eqref{Mneq0_r1_r2_r12} in~\eqref{Mneq0_eta2dot} and 
equating the real and imaginary parts of both sides, we get
\begin{align}
\begin{aligned}
\dot{u}&=vf(u,v)\left[\frac{\Gamma_1+\Gamma_2(u^2+v^2)}{2M\pi(u^2+v^2)\left((u-1)^2+v^2\right)}\right],\\
\dot{v}&=-g(u,v)\left[\frac{\Gamma_1+\Gamma_2(u^2+v^2)}{2M\pi(u^2+v^2)\left((u-1)^2+v^2\right)}\right],
\end{aligned}
\label{Mneq0_udot_vdot}
\end{align}
where $
f(u,v)=\Gamma_0\left((u-1)^2+v^2\right)(u^2+v^2-1)+\left(-\Gamma_1+\Gamma_2(1-2u)\right)(u^2+v^2)$,
and $g(u,v)=uf(u,v)+(u^2+v^2)\left(\Gamma_1+\Gamma_2(u^2+v^2)\right)$. 
It is worth noticing that~\eqref{Mneq0_udot_vdot} is invariant under the transformation $t\to -t$ and $v\to -v$, and, therefore, is a reversible system~\cite{SS2000,SW2003}. 
Thus, for any trajectory in the positive $v$-plane there is a trajectory in the negative $v$-plane, which are mirror images of each other.

In the following sections, we will discuss the equilibrium solutions and trajectories of the dynamical system~\eqref{Mneq0_udot_vdot}. 

\subsubsection{Equilibrium solutions}
\label{subsubsec:Mneq0_equilibrium}

To find the equilibrium solutions, we look for points $(u,v)\in\mathbb{R}^2\backslash\{(0,0),(1,0)\}$  satisfying $\dot{u}=0=\dot{v}$ in~\eqref{Mneq0_udot_vdot}. 

For $M\neq 0$, the term $\Gamma_1+\Gamma_2(u^2+v^2)$ must be non-zero for all time. Consequently, by equating $\dot{u}$ and $\dot{v}$ to zero, one obtains $vf(u,v)=0$ and $g(u,v)=0$ respectively. The term $f(u,v)$ cannot be zero as it implies $g(u,v)=(u^2+v^2)\Big(\Gamma_1+\Gamma_2(u^2+v^2)\Big)$,  which cannot be zero. Therefore, $v$ must be equal to zero and $u$ must satisfy the polynomial $g(u,0)=0$. Recall that $(0,0)$ and $(1,0)$ are singularity points, and they cannot be equilibrium points. Thus, the factor $u(u-1)$ in the expression of $g(u,0)$ cannot be zero, and we finally end up with a cubic polynomial
\begin{equation}
\label{Mneq0_cubic_polynomial}
p(u)=u^3-(1+\alpha_2)u^2-(1+\alpha_1)u+1,\qquad
\alpha_1=\Gamma_1/\Gamma_0, \quad \alpha_2=\Gamma_2/\Gamma_0,
\end{equation} 
whose real roots correspond to the location of the equilibrium points on the $u$-axis. Note that as all the equilibrium points lie on the $u$-axis ($v=0$), collinear fixed configurations are the only possible type of self-similar vortex evolutions in the case of $M\neq0$ (see lemmas~\ref{lemma_self_similarity} and~\ref{lemma_fixed_configuration}). 

Since the diagonal entries of the Jacobian matrix associated with~\eqref{Mneq0_udot_vdot} are zero at the $u$-axis, the equilibrium points of the linearized system are either centers or saddles. As the system~\eqref{Mneq0_udot_vdot} is reversible, it follows that equilibrium points of the original non-linear system are also either centers or saddles (see, e.g.~\cite{SS2000,PL2001}).
Let us now explore various kinds of trajectories possible in the $(u,v)$ phase plane, and the corresponding physical implications about the vortex motion.

\subsubsection{Trajectories}

Writing Hamiltonian in terms $u$ and $v$ by substituting $r_1$, $r_2$, and $r_{12}$ from~\eqref{Mneq0_r1_r2_r12} into \eqref{cons_energy}, we get
\begin{multline}
\Psi(u,v):=
\Gamma_0\Gamma_1\log\left|\frac{1}{\Gamma_1+\Gamma_2(u^2+v^2)}\right|+\Gamma_0\Gamma_2\log\left|\frac{(u^2+v^2)}{\Gamma_1+\Gamma_2(u^2+v^2)}\right|\\+\Gamma_1\Gamma_2\log\left|\frac{(u-1)^2+v^2}{\Gamma_1+\Gamma_2(u^2+v^2)}\right|=\text{constant.}
\label{Mneq0_levelcurve}
\end{multline}
Thus, any trajectory in the $(u,v)$ phase plane can be described as a level curve given by ~\eqref{Mneq0_levelcurve}, with the constant term determinable from the initial conditions. Note that, in~\eqref{Mneq0_levelcurve}, the constant term is finite, and $v^2$ dependency is a direct consequence of the reversibility of system~\eqref{Mneq0_udot_vdot}. 
%
%

From~\eqref{r1_sqr_dot_uv_eqn}--\eqref{r2_sqr_dot_uv_eqn},  we observe that the extrema of $r_1$ and $r_2$ must lie on the $u$-axis. Consequently, $r_1$ (and $r_2$) attains its maximum and minimum in finite time if and only if  the trajectory in the $(u,v)$ phase plane is periodic and closed due to the reversibility of~\eqref{Mneq0_udot_vdot}. This indicates that the boundedness  of the vortices and periodicity of the variable $\eta_2=(u,v)$ might be  interdependent, which we shall investigate later.

In the following lemma, we look at the boundedness of the $(u,v)$ phase plane distance functions $|\eta_2|^2=u^2+v^2$ and $|\eta_2-\eta_1|^2=(u-1)^2+v^2$ for $t\in \mathbb{R}$.
\begin{lemma}\label{lemma_Mneq0_uv_boundedness}
For any $(u,v)$ phase plane trajectory $\eta_2(t)=(u(t),v(t)), t\in\mathbb{R}$ (i) $u^2+v^2$ is bounded away from  zero as well as bounded above; (ii) if $\Gamma_1+\Gamma_2\neq 0$ then $(u-1)^2+v^2$ is also bounded away from zero.
\end{lemma}
\begin{proof}
$(i)$ Substitute $u=r\cos\theta$ and $v=r\sin\theta$ in~\eqref{Mneq0_levelcurve} and consider the limit $r\to 0$ and $r\to\infty$. In both cases, the left-hand side of~\eqref{Mneq0_levelcurve} is not finite, which is a contradiction. 
$(ii)$ As $\eta_2\to\eta_1$, the left-hand side of the~\eqref{Mneq0_levelcurve} is not finite. Since expression~\eqref{Mneq0_levelcurve} must always yield a finite constant, we conclude that $(u-1)^2+v^2$ is bounded away from zero.
\end{proof}
Hence whatever be the initial conditions, a trajectory in $(u,v)$ phase plane is always bounded away from the singularity point $\eta_0=(0,0)$. The same can be  said about the singularity point $\eta_1=(1,0)$, if the free vortices are not of equal counter-rotating type, i.e., when $\Gamma_1+\Gamma_2\neq 0$.

We shall now characterize the closed periodic orbits in the $(u,v)$ phase plane. 
\begin{theorem}\label{theorem_Mneq0_periodic_orbits}
A $(u,v)$ phase plane  trajectory is bounded away from the equilibrium points and the singularity point $\eta_1=(1,0)$ for $t\geq0$ (or $t\leq 0$) if and only if it is a closed orbit. 
\end{theorem} 
\begin{proof}
Let  $\eta_2(t)=(u(t),v(t))$ be a trajectory bounded away from the equilibrium points and the singularity point $\eta_1$.  Consider the set $S_1=\{\eta_2(t)|t\geq0\}$. For $\epsilon>0$, define the set $S_2=\{(u,v)|\inf_{(u',v')\in S_1} (u-u')^2+(v-v'^2)\leq \epsilon\}$. The set $S_2$ is a compact set in $\mathbb{R}^2$ that contains $S_1$. For sufficiently small $\epsilon$, $S_2$ does not contain any of the equilibrium or singularity points. Hence from the Poincar\'{e}-Bendixson theorem either $S_1$ is a closed orbit, or it spirals towards a limit cycle. Since a conservative system in $\mathbb{R}^2$ cannot have a limit cycle, we conclude that $\eta_2$ is periodic in time.

  Conversely, if $\eta_2$ is closed and periodic, $S_1$ must be bounded away from the equilibrium points and the singularity point $\eta_1=(1,0)$. Note that if we replace $t\geq0$ by $t\leq0$ in the theorem, the result still holds.
\end{proof}
\begin{corollary}
\label{corollary_Mneq0_classsification}
Any non-equilibrium $(u,v)$ phase plane trajectory bounded away from the singularity point $\eta_1=(1,0)$ for $t\geq 0$ ($t\leq0$) is either a stable (unstable) separatrix of a saddle equilibrium point, or a periodic trajectory. 
\end{corollary}
\begin{proof}
Directly follows from theorem~\ref{theorem_Mneq0_periodic_orbits},
and the fact that equilibrium points of system~\eqref{Mneq0_udot_vdot} are either centers or saddles (see Sec.~\ref{subsubsec:Mneq0_equilibrium}).
\end{proof}
\begin{remark}
\label{remark_sign_convention}
For $\Gamma_1\Gamma_2<0$, we observe the following. Since $M\neq 0$, the trajectories in the $(u,v)$ phase planes are contained in either the interior or the exterior of the circle given by  $u^2+v^2=-\Gamma_1/\Gamma_2$. We can always reduce the $(u,v)$ phase plane trajectory to that of former type by appropriately indexing the free vortices. In other words, it is enough to study the $(u,v)$ phase plane dynamics for $u^2+v^2\leq -\Gamma_1/\Gamma_2$. 
\end{remark}
Next, we shall show that for $M\neq0$ case, the vortex motion is unbounded only if free vortices are of equal counter-rotating type, i.e., $\Gamma_1+\Gamma_2=0$.
\begin{lemma}
\label{lemma_Mneq0_g1pg2neq0case}
If $\Gamma_1+\Gamma_2\neq 0$, then the vortex motion is bounded, i.e., $r_1$ and $r_2$ are bounded above for all time.
\begin{proof}
We will consider two cases (i) $\Gamma_1\Gamma_2>0$ and (ii) $\Gamma_1\Gamma_2<0$.

(i) When $\Gamma_1\Gamma_2>0$, \eqref{cons_ang} represents an ellipse, 
and therefore $r_1$ and $r_2$ are bounded above. 

(ii) When $\Gamma_1\Gamma_2<0$,
it is seen from~\eqref{cons_ang} that $r_1$ and $r_2$ can only tend to infinity simultaneously.
 Hence it suffices to show that $r_1$ is bounded above. Consider a trajectory $\eta_2(t)$ in $(u,v)$ phase plane for $t>0$ ($t<0$ case follows similarly) with $\eta_2|_{t=0}=(u_0,v_0)$. WLOG, we assume that $u_0^2+v_0^2<-\Gamma_1/\Gamma_2$. The trajectory $\eta_2(t)=\left(u(t),v(t)\right)$ lies in the interior of  the circle $C_\kappa=\{(u,v)|u^2+v^2=\kappa^2\}$, where $\kappa=\sqrt{-\Gamma_1/\Gamma_2}>0$;  the elements of $C_\kappa$ correspond to the $M=0$ case. From corollary~\ref{corollary_Mneq0_classsification} 
and lemma~\ref{lemma_Mneq0_uv_boundedness}, it follows that $\eta_2$ is either a periodic orbit or a stable separatrix of a saddle point. Both these cases correspond to bounded $r_1$,  except when $\eta_2$ is a stable separatrix of a saddle equilibrium $(\tilde{u},0)$ with $|\tilde{u}|=\kappa$. 
However, such a trajectory is not possible from the continuity of solutions. Since from expression~\eqref{Mneq0_r1_r2_r12} it would imply, $r_1$ must tend to infinity as $t$ tends to infinity, and by continuity, we shall have $r_1(\tilde{u},0)$ to be infinite. This is a contradiction to the fact that points on the $u$-axis correspond to a finite $r_1(u,v)$ value for $M=0$ case when $\Gamma_1+\Gamma_2\neq 0$. So that the trajectory $\eta_2$ is bounded away from $C_\kappa$ when $\Gamma_1+\Gamma_2\neq 0$, and the vortex motion is bounded in all cases.
\end{proof}
\end{lemma}
We now show that free vortices always stay close to each other.

\begin{lemma}
\label{lemma_Mneq0_r12_bounded}
$r_{12}$ is bounded away from zero and bounded above for all time.
\begin{proof}

If the vortex motion is bounded, then combined with lemma~\ref{lemma_Mneq0_lowerbound_r1r2}, the vortex distances  $r_1$ and $r_2$ are bounded away from zero and bounded above. It follows from~\eqref{cons_energy} that $\Gamma_1\Gamma_2\log r_{12}$ is bounded. Consequently,  $r_{12}$ must be bounded on both sides.

If the vortex motion is unbounded then from lemma~\ref{lemma_Mneq0_g1pg2neq0case}, $\Gamma_1+\Gamma_2$ must be equal to zero and~\eqref{cons_energy} simplifies to
\begin{equation}
-4\pi H=\Gamma_0\Gamma_2\log(u^2+v^2)-\Gamma_2^2\log(r_{12}^2).
\end{equation}
In the above equation, $u^2+v^2$ term is bounded on both sides (see lemma~\ref{lemma_Mneq0_uv_boundedness}), and the right-hand side is a finite constant. Therefore, $r_{12}$ is bounded away from zero and bounded above.
\end{proof}
\end{lemma}

In the following lemma, we physically characterize the $(u,v)$ phase plane trajectories that converge to the singularity point $\eta_1=(1,0)$.
\begin{lemma}\label{lemma_Mneq0_eta1_convergence}
The vortex motion is unbounded if and only if $(u-1)^2+v^2 $ tends to zero, i.e., the $(u,v)$ trajectory tends to the singularity point $\eta_1=(1,0)$.
\end{lemma}
\begin{proof} 
Follows from~\ref{lemma_Mneq0_r12_bounded} and~\eqref{r2byr1_r12byr1}.

%
\end{proof}
From lemma~\ref{lemma_Mneq0_uv_boundedness}, we know that all $(u,v)$ phase plane trajectories are bounded away from the singularity point $\eta_0=(0,0)$. However, this is not the case for the second singularity point $\eta_1=(1,0)$. We may have a $(u,v)$ trajectory converging to $\eta_1$ in the equal counter-rotating case, and lemma~\ref{lemma_Mneq0_eta1_convergence} states that this physically corresponds to an unbounded vortex motion, and vice-versa. The existence of initial conditions leading to such trajectories is explained in theorem~\ref{theorem_Mneq0_ultimatum}.  

\begin{lemma} 
For both $t>0$ and $t<0$ a non-equilibrium $(u,v)$ phase plane trajectory must either (i) intersect the $u$-axis in finite time, or (ii) converge to a saddle equilibrium point or the singularity point $\eta_1=(1,0)$.
\label{lemma_Mneq0_periodic_or_asymptiotic}
\begin{proof} 
Let $\eta_2(t)=(u(t),v(t))$ be a non-equilibrium trajectory.
WLOG we may assume $\eta_2|_{t=0}=(u_0,v_0)$ and $v_0\neq 0$.
Let us only look at the case $t>0$, as similar lines of arguments can be given for $t<0$. If the vortex motion is unbounded for $t>0$, then the corresponding $(u,v)$ trajectory must tend to the singularity point $\eta_1=(1,0)$ as $t$ tends to infinity (see lemma~\ref{lemma_Mneq0_eta1_convergence}). Now suppose that $r_1$ is bounded and that for $t>0$, the trajectory $\eta_2$ does not intersect the $u$-axis. Thus, $\eta_2$ cannot be a periodic trajectory and from corollary~\ref{corollary_Mneq0_classsification}, it must tend to a saddle equilibrium point as $t$ tends to infinity.
\end{proof}
\end{lemma}
Next, we show that the singularity point $\eta_0=(0,0)$ has an index $+1$, and hence there is always a region of closed trajectories surrounding the origin in the $(u,v)$ phase plane.

\begin{lemma}
The origin has an index $+1$.
\begin{proof}
We shall show that there exists a closed trajectory which contains the origin but none of the equilibrium points or the singularity point $\eta_1$. Consider the open ball $B_d(0)=\{(u,v)\big|u^2+v^2<d^2\}$, where $d=1/2\,\min\{|u|\,\big|p(u)=0\}$. By construction, $B_d(0)$ does not contain any of the equilibrium points or the singularity point $(1,0)$. We shall try to find a point $(u_0,v_0)\in B_d(0)$  such that the unique trajectory $\eta_2$ that passes through $(u_0,v_0)$ is contained in $B_d(0)$ for all time. Let $h(u,v)=u^2+v^2$, we have $\dot{h}=-v\Big(\Gamma_1+\Gamma_2(u^2+v^2)\Big)^2/4M\pi\Big((u-1)^2+v^2\Big)$. Hence points on the $u$-axis are either a minimum or maximum for the function $h$. It can be verified that the sign of the second derivative $\ddot{h}$ depends only on the sign of $\dot{v}$. The expression for $\dot{v}$ evaluated on the $u$-axis is
\begin{equation*}
\dot{v}|_{(u,0)}=-\frac{\Gamma_0(\Gamma_1+\Gamma_2u^2)}{2M\pi u^2(u-1)^2}\times p(u)\times u\times (u-1).
\end{equation*}
In the above expression, term  $\Gamma_0(\Gamma_1+\Gamma_2u^2)/2M\pi u^2(u-1)^2$  has a constant sign irrespective of the sign of $u$. In addition, $p(u)>0$ for any $u$ in a sufficiently small neighbourhood of the origin from the continuity of $p$ and the fact that $p(0)=1$. Hence by appropriately choosing $(u,0)$  negative or positive from a sufficiently close neighbourhood of the origin, we can make sure that $\ddot{h}<0$, a maximum for the function $h$. Therefore, a trajectory $\eta_2$ that originates at this maximum point of $h$ would be contained in $B_d(0)$ for all time. From theorem~\ref{theorem_Mneq0_periodic_orbits}, it follows that $\eta_2$ is a closed trajectory. Since any closed trajectory should contain at least one equilibrium or singularity point, the origin must be in the interior of this trajectory and therefore has an index $+1$.
\end{proof} 
\end{lemma}
In the following lemma, we show that if the free vortices are of equal counter-rotating type, then the corresponding  $(u,v)$ phase plane reduces to a disc of radius one centered at the origin, and it contains precisely one (saddle) equilibrium point.

\begin{lemma}
If $\Gamma_1+\Gamma_2=0$, then there exists only one equilibrium point in $(-1,1)$, and it is a saddle.
\end{lemma}
\begin{proof}
From remark~\ref{remark_sign_convention}, 
we have $u^2+v^2\leq 1$.
Recall from Sec.~\ref{subsubsec:Mneq0_equilibrium}, all equilibrium points reside on $u$-axis and the $u$ must be the root of $p(u)=u^3-(1+\alpha_2)u^2-(1-\alpha_1)u+1$, where $\alpha_1=\Gamma_1/\Gamma_0$ and $\alpha_2=\Gamma_2/\Gamma_0$. Since $\alpha_1+\alpha_2=0$, one can factorize the polynomial $p$ as 
$p(u)=(u-1)q(u)$,
where $q(u)=u^2-\alpha_2u-1$. Since $u$ cannot be one, this would mean that the $u$ coordinate of the  equilibrium point must be a root of 
$q$. Let $u_1$ and $u_2$ denote the two roots of $q$. 
As $q(-1)=\alpha_2$, and $q(1)=-\alpha_2$, by continuity at least one of these two roots lies in $(-1,1)$. Since $u_1u_2=-1$, the second root cannot be in $(-1,1)$. Hence, there is exactly one equilibrium point in the region $u^2+v^2\leq 1$, and  we shall denote this unique equilibrium point by $(u_s,0)$. The linearized system has eigenvalues given by
\begin{equation*}
\lambda^\pm=\pm \frac{\Gamma_2 (1+u_s)\sqrt{\Gamma_2^2(-1+u_s)^2(u_s+u_s^3)}}{2M\pi(-1+u_s)u_s^{3/2}}.
\end{equation*}
The product $\lambda^+\lambda^-=-\Gamma_2^2(1+u_s)^2(u_s+u_s^3)/4M\pi^2u_s^3<0$ irrespective of the value of $u_s$. Hence it is a saddle equilibrium point.
\end{proof}

The following theorem characterizes the initial conditions with respect to the vortex boundedness in the equal counter-rotating free vortex pair case.

\begin{theorem}
Let $\Gamma_1+\Gamma_2=0$ and the vortices be indexed such that $|z_2|<|z_1|$ at $t=0$. The necessary and sufficient condition for vortex entrapment is that the initial point $(u_0,v_0)=(u,v)|_{t=0}$ lies in the interior of the curve given by $\Psi(u,v)=\Psi(u_s,0)$ that encloses the origin, where $u_s$  is the unique root of the quadratic polynomial  $u^2- (\Gamma_2/\Gamma_0) u-1$
in the interval $(-1,1)$. 
\label{theorem_Mneq0_ultimatum}
\begin{proof}

There are exactly two trajectories that approach and originate from a saddle 
point. Let us look at the two unstable separatrices that originate from the unique saddle. For one of these trajectories, the saddle is a point of maximum for the inter-vortex distance $r_1$ and a minimum for the other. 

The first trajectory corresponds to a bounded motion, and therefore must be  bounded away from the singularity point $(1,0)$ (see lemma~\ref{lemma_Mneq0_eta1_convergence}). Hence from lemma~\ref{lemma_Mneq0_periodic_or_asymptiotic}, this non-equilibrium trajectory either intersect a point on the $u$-axis in finite time or tend to a saddle equilibrium point asymptotically for $t>0$. Since there is only one saddle equilibrium point, the latter case is not possible. Hence the unstable separatrix 
must intersect the $u$-axis in finite time. Coupled with the reversibility of the system, this gives us a homoclinic orbit. Since trajectories in the interior of this homoclinic orbit are bounded away from saddle and $(1,0)$ point, they are closed trajectories. Since closed trajectories must contain equilibrium points or singularities of total index $+1$, this can only happen if the origin is contained in the interior of the homoclinic orbit under consideration. 

\begin{figure}
\centering
\subfigure[]{
\includegraphics[height=1.7in]{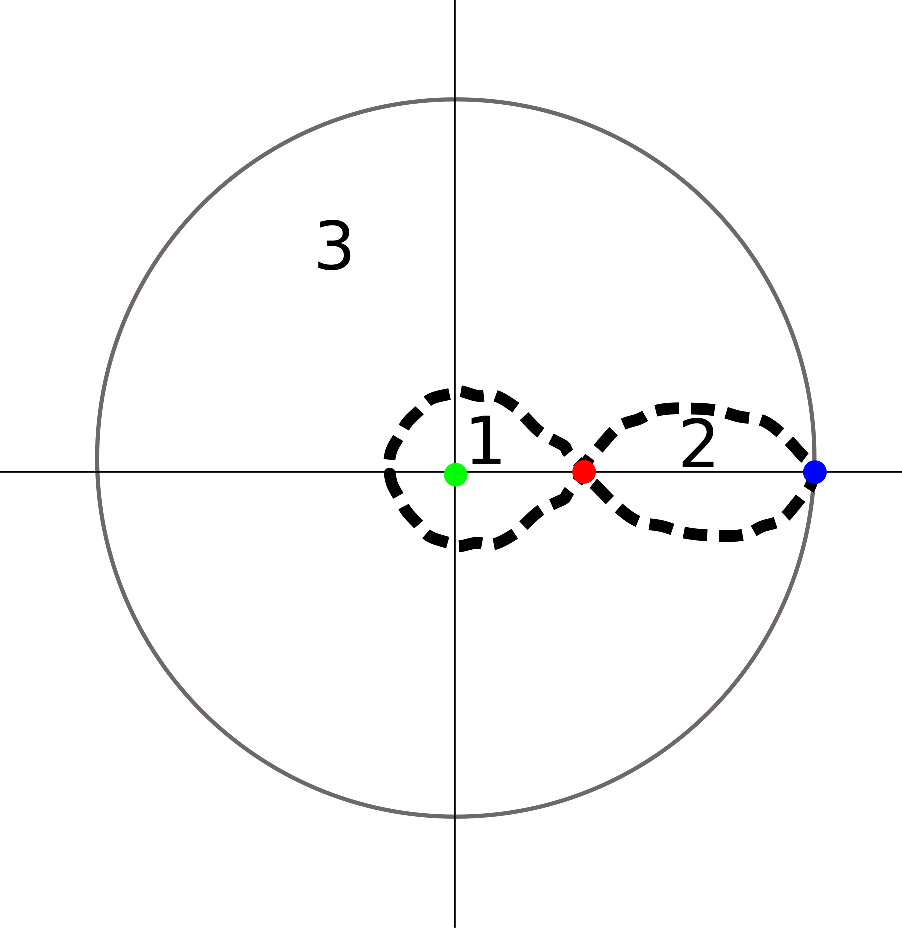}}\qquad\qquad
\subfigure[]{\includegraphics[height=1.7in]{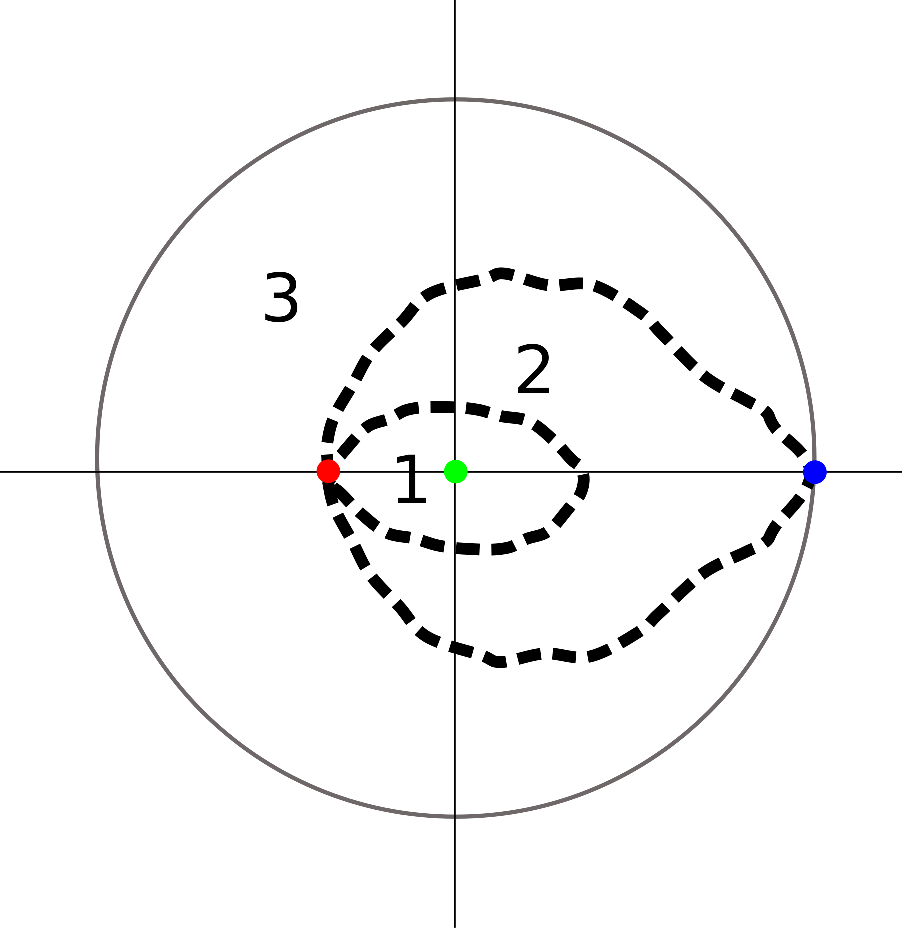}}
\caption{\small{Schematic showing the two types of $(u,v)$ phase plane ($|z|<1$) separatrices (dashed lines) when the unique saddle on the $u$-axis (red dot) lies in (a) $(0,1)$, (b) $(-1,0)$. The green and blue dots in the panels correspond to the singularity points $\eta_0=(0,0)$ and $\eta_1=(1,0)$, respectively.
}}
\label{fig_Mneq0_three_regions}
\end{figure}

Let us now look at the second unstable separatrix that has the saddle as a minimum for $r_1$. From lemma~\ref{lemma_Mneq0_periodic_or_asymptiotic}, it must either intersect the $u$-axis in finite time or tend to $(1,0)$ point. The first case cannot happen as that would mean a region in the phase plane having closed trajectories but does not contain points having an index sum to $+1$. The second unstable separatrix trajectory thus tends to $(1,0)$ point. 

Overall the two unstable separatrices 
subdivide the phase plane $|z|<1$ into three regions (see figure~\ref{fig_Mneq0_three_regions}).  Trajectories in  region one are bounded away from $(1,0)$ and the saddle point. Therefore, all trajectories in region one are closed and periodic. Since trajectories in regions two and three cannot have the origin in their interior, none of them are closed. As these trajectories are also bounded away from the saddle point, from corollary~\ref{corollary_Mneq0_classsification}, they must tend to the singularity point $(1,0)$ from one side and have a $u$-axis intersection in finite time on the other side. Coupled with reversibility, we see that eventually all trajectories in regions~2 and~3 must tend to the singularity point $(1,0)$, which corresponds to an unbounded vortex motion. 
\end{proof}
\end{theorem}

In the asymmetric case $M \neq 0$, the motion of the free vortices $\mathcal{V}_1$ and $\mathcal{V}_2$ are always bounded in a neighbourhood of the fixed vortex $\mathcal{V}_0$, if they are not of equal counter-rotating type (see lemma~\ref{lemma_Mneq0_g1pg2neq0case}). Theorem~\ref{theorem_Mneq0_ultimatum} gives a necessary and sufficient condition for a bounded vortex motion in the counter-rotating case. Given the initial conditions, we index the vortices such that $|z_2|_{t=0}<|z_1|_{t=0}$. If the quotient $z_2/z_1|_{t=0}$ lies in the interior of the region given by $\Psi(u,v)=\Psi(u_s,0)$ that contains the origin (region~1 in figure~\ref{fig_Mneq0_three_regions}), then the vortex motion is bounded with periodic inter-vortex distances. Otherwise, the vortex motion is unbounded. 

\subsection{Examples for \texorpdfstring{$M\neq 0$}{M neq 0} case}
\label{subsec:Examples_Mneq0}

In this section we shall illustrate results by considering two physically important special cases of circulations, namely, the equal vortices and equal counter-rotating vortex pair.

\subsubsection{Equal vortices (\texorpdfstring{$\Gamma_0=\Gamma_1=\Gamma_2\neq0$}{Gamma0=Gamma1=Gamma2 neq 0})}

\begin{figure}[htbp!]
\centering
\includegraphics[height=2.5in]{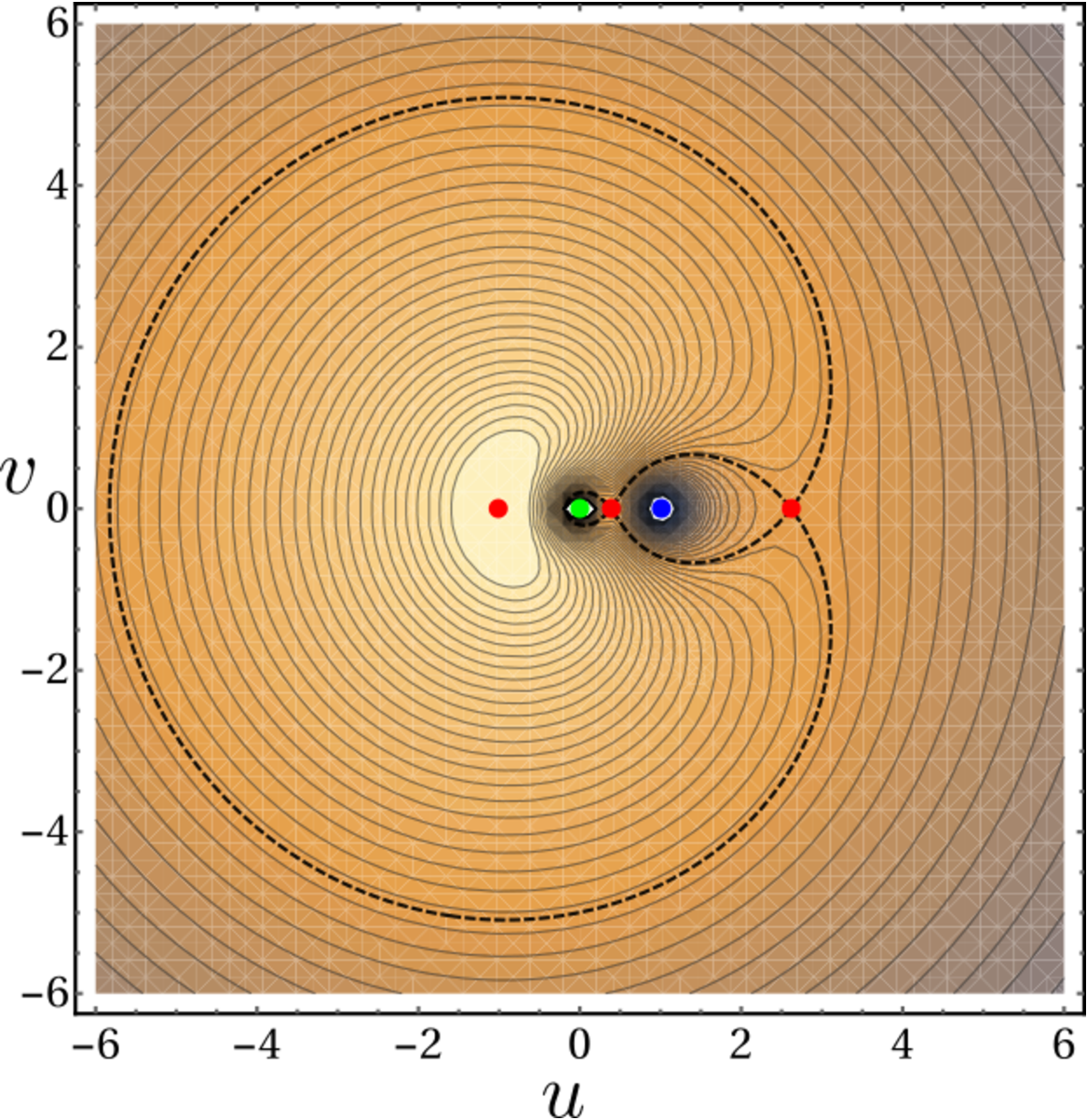}
\caption{Contours of equal-vortex case ($\Gamma_0=\Gamma_1=\Gamma_2\neq0$)}
\label{fig_Mneq0_equal_contour}
\end{figure}
When the vortices are equal,~\eqref{Mneq0_levelcurve} simplifies to
\begin{equation}\label{Mneq0_equal_levelcurve}
\frac{\left(u^2+v^2\right)\left((u-1)^2+v^2\right)}{\left(1+u^2+v^2\right)^3}=\text{ constant.}
\end{equation}
Figure~\ref{fig_Mneq0_equal_contour} shows the contours of~\eqref{Mneq0_equal_levelcurve} representing the trajectories of 
$\eta_2=(u(t),v(t))$ for different initial conditions. 
The exact location of the equilibrium points, see~\eqref{Mneq0_cubic_polynomial}, are found by solving the cubic equation $p(u)=u^3-2u^2-2u+1=0$.
The roots are given by
 $u_1=-1$, $u_2=(3-\sqrt{5})/2\approx 0.381966$, and $u_3=(3+\sqrt{5})/2\approx2.61803$. From figure~\ref{fig_Mneq0_equal_contour}, it is evident that the equilibrium points at $(u_1,0)$ is a center (leftmost red dot), and at $(u_2,0)$ and $(u_3,0)$ are saddles (other two red dots), as discussed at the end of Sec.~\ref{subsubsec:Mneq0_equilibrium}. It is  seen that the  trajectories are either (i) equilibrium points (red dots), (ii) separatrices of the saddle equilibrium points (black dashed lines), and (iii) closed periodic trajectories (black continuous lines), which corresponds to (i) a fixed configuration of vortices (see lemma~\ref{lemma_fixed_configuration} and figure~\ref{fig_Mneq0_equal_fixed_conf}), (ii) vortex motion that asymptotically converges to an unstable fixed configuration (see figure~\ref{fig_Mneq0_equal_aperiodic}), and (iii) vortex motion in which inter-vortex distances are periodic (see figure~\ref{fig_Mneq0_equal_periodic}), respectively. These three cases are illustrated below.

\begin{enumerate}[label=(\roman*)]
\item
{Fixed configuration:}
\begin{figure}[htb!]
\centering
\subfigure[]{
\includegraphics[width=0.31\textwidth,height=0.31\textwidth]{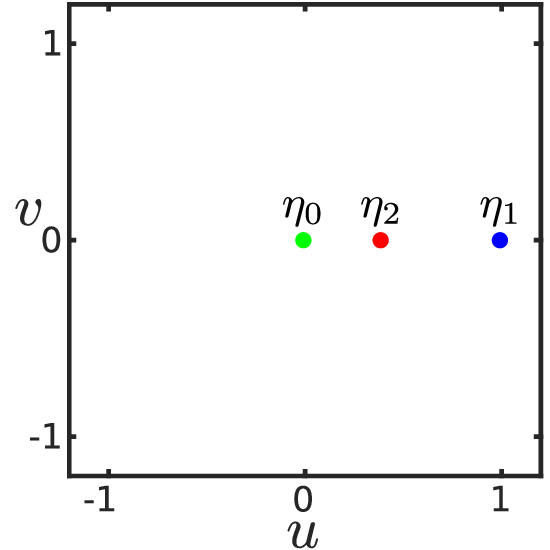}
}
\subfigure[]{
\includegraphics[width=0.31\textwidth,,height=0.31\textwidth]{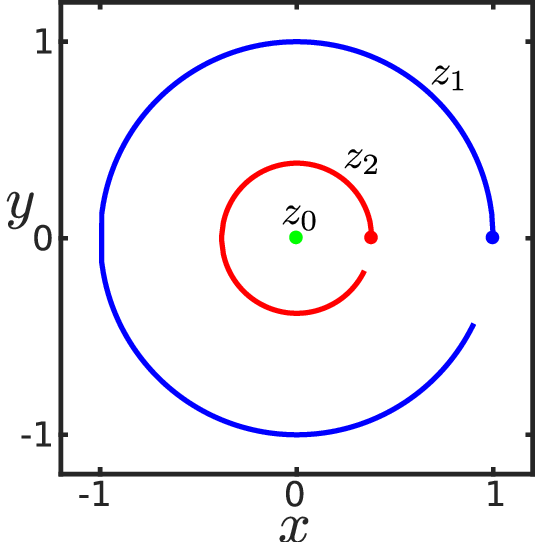}
}
\subfigure[]{
\includegraphics[width=0.31\textwidth,height=0.31\textwidth]{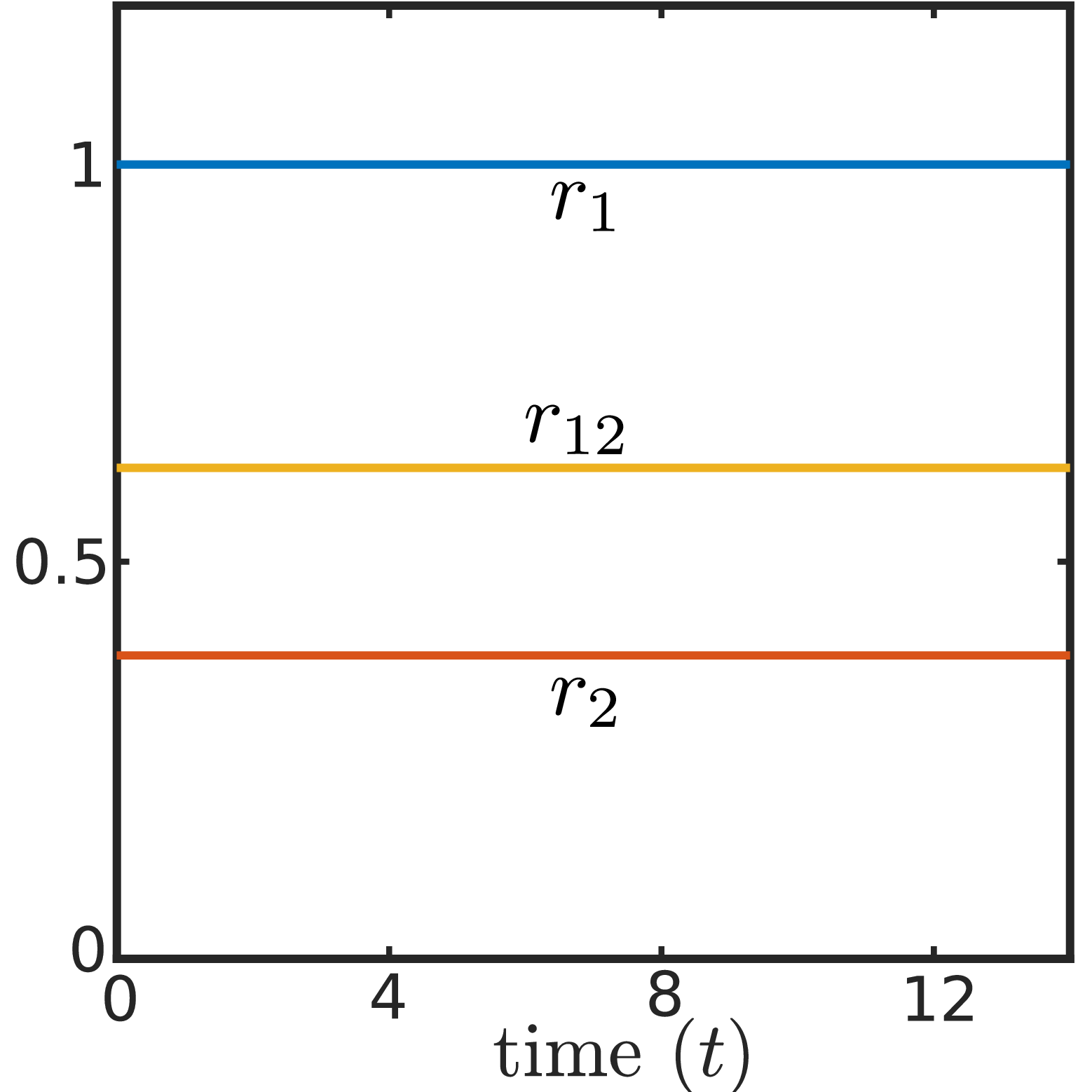}
}
\caption{\small{Fixed configuration ($\Gamma_0=\Gamma_1=\Gamma_2=1$) of vortices. The positions of vortices $\mathcal{V}_0$, $\mathcal{V}_1$, and $\mathcal{V}_2$ are marked by green, blue, and red, respectively, in (a) $(u,v)$ plane and (b) $(x,y)$ plane. (c) Variation of inter-vortex distances with time.}}
\label{fig_Mneq0_equal_fixed_conf}
\end{figure}

Since from lemma~\ref{lemma_fixed_configuration}, any initial condition that leads to a fixed configuration of vortices corresponds to an equilibrium solution on the $u$-axis, we consider $z_1|_{t=0}=1, z_2|_{t=0}=(3-\sqrt{5})/2$ so that $\eta_2|_{t=0}$ corresponds to one of the two saddle equilibrium points described earlier. The system~\eqref{z0_z1_z2_dot} is numerically integrated till $t=14$ using the fourth-order Runge-Kutta method. Plotting  $\eta_2=z_2/z_1$  yields an equilibrium trajectory (marked in red) as in figure~\ref{fig_Mneq0_equal_fixed_conf}(a). From figure~\ref{fig_Mneq0_equal_fixed_conf}(b), we see that the actual vortex motion consists of vortices $\mathcal{V}_1$ and $\mathcal{V}_2$ revolving around the fixed vortex $\mathcal{V}_0$ in circular orbits with the same angular velocity so that they remain collinear at any point of time. Moreover, the inter-vortex distances are constants as evident from figure~\ref{fig_Mneq0_equal_fixed_conf}(c).  

\item
{Aperiodic case}

\begin{figure}[htbp!]
\centering
\subfigure[]{
\includegraphics[width=0.31\textwidth,height=0.31\textwidth]{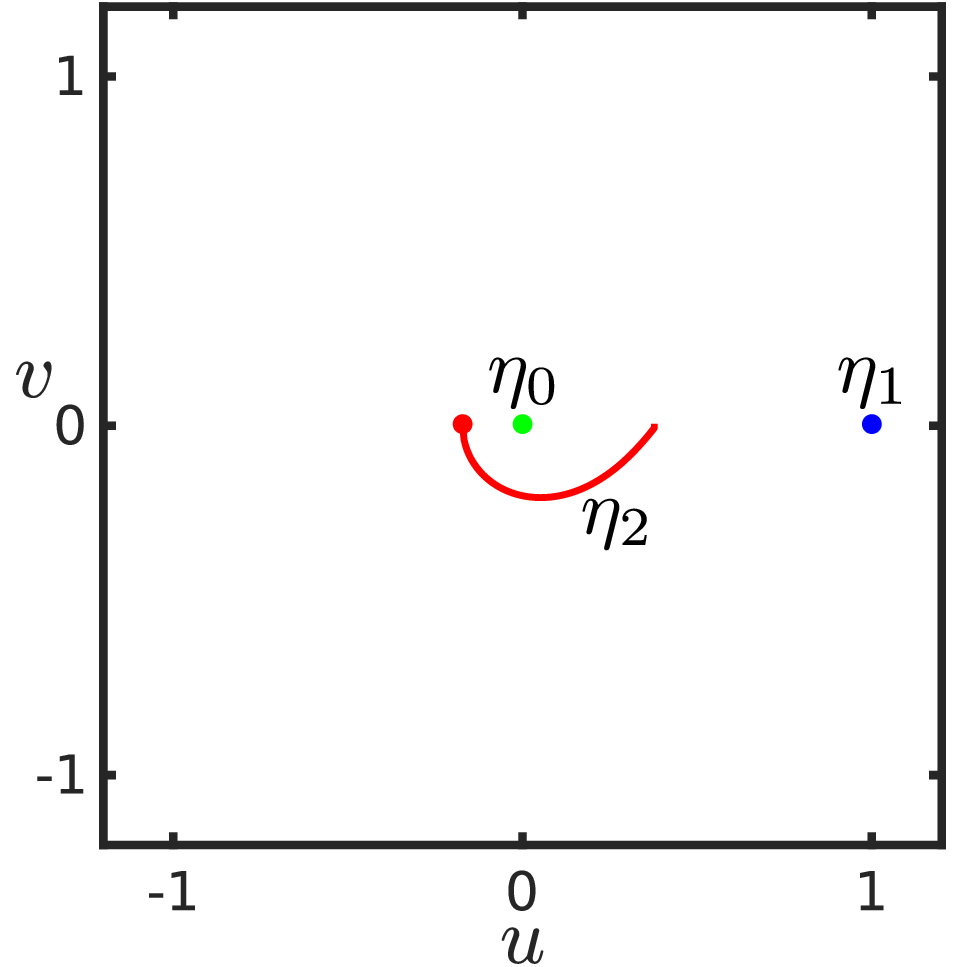}
}
\subfigure[]{
\includegraphics[width=0.31\textwidth,height=0.31\textwidth]{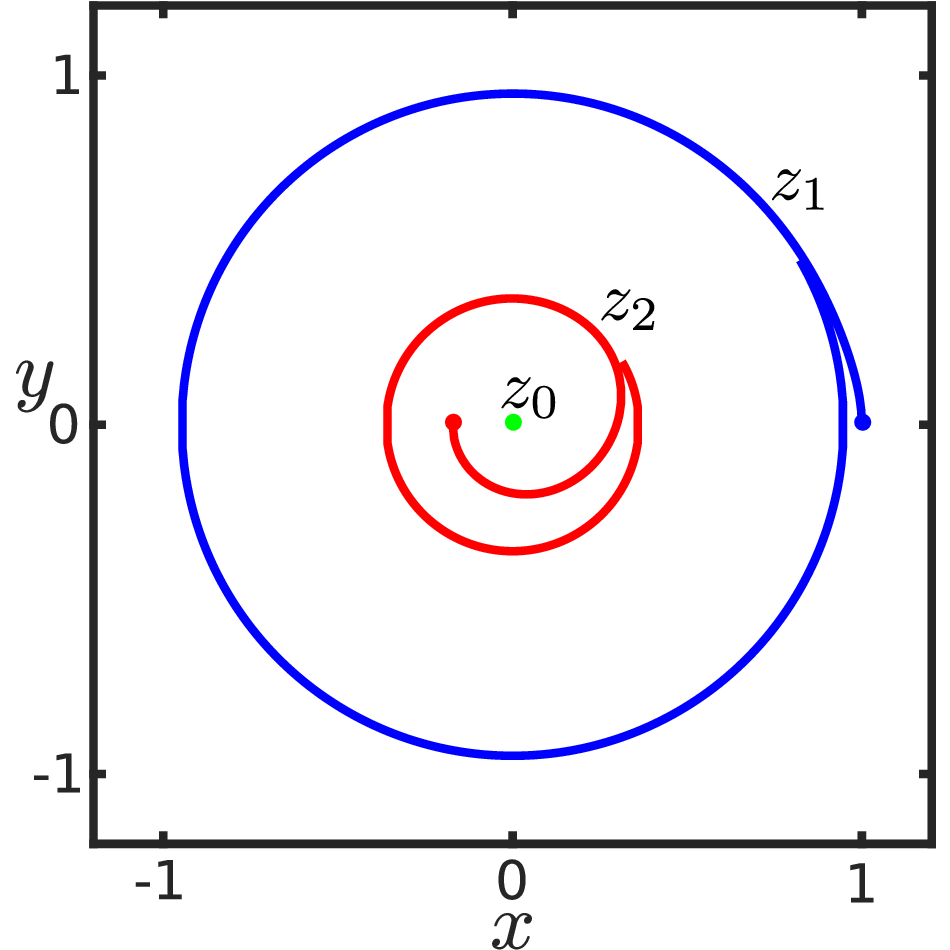}
}
\subfigure[]{
\includegraphics[width=0.31\textwidth,height=0.31\textwidth]{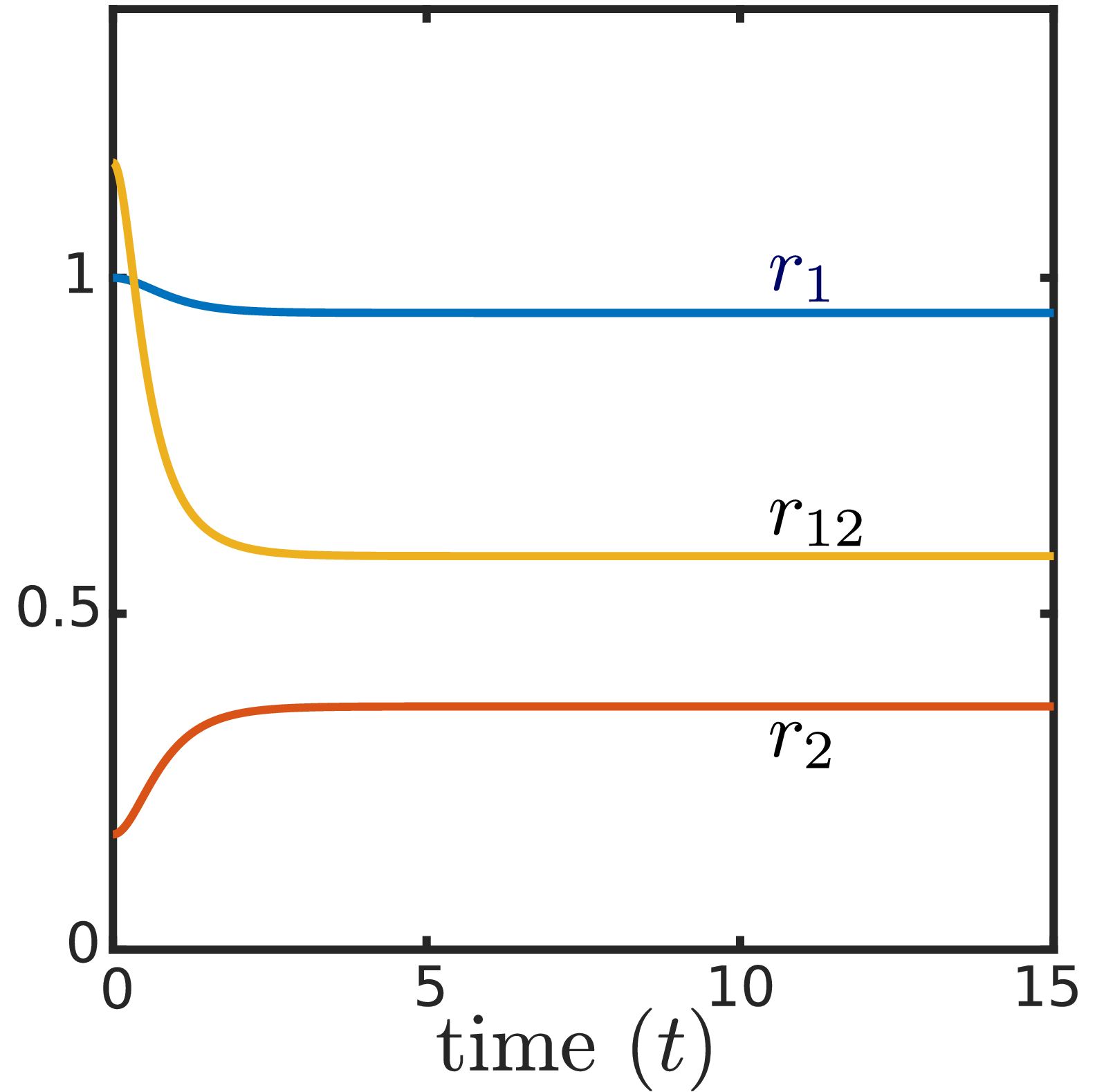}
}
\caption{\small{
Same as figure~\ref{fig_Mneq0_equal_fixed_conf} but for the aperiodic case. 
}}
\label{fig_Mneq0_equal_aperiodic}
\end{figure}

Any initial condition which has the quotient $z_2/z_1|_{t=0}$ lying on the separatrices would asymptotically approach to a saddle equilibrium point in the $(u,v)$ phase plane. Physically this would mean that the vortex trajectories would look more and more like a fixed configuration for larger time scale. To illustrate this, we consider the set of initial conditions, $z_1|_{t=0}=1$ and $z_2|_{t=0}\approx -0.171573$, so that $z_2/z_1|_{t=0}$ is a non-equilibrium point on the separatrices. A numerical plot of $\eta_2=z_2/z_1$ gives us a $(u,v)$ phase plane trajectory, which tends to the saddle equilibrium point situated at $\left((3-\sqrt{5})/2,0\right)$. The vortex trajectories [see figure~\ref{fig_Mneq0_equal_aperiodic}(b)] are found to be the one in which the vortices approach the collinear circular orbits described earlier in figure~\ref{fig_Mneq0_equal_fixed_conf}(b). The inter-vortex distances also tend to a constant limiting value as in figure~\ref{fig_Mneq0_equal_aperiodic}(c).

\item
{Periodic case}


From lemma~\ref{lemma_Mneq0_uv_boundedness} and corollary~\ref{corollary_Mneq0_classsification}, all the initial conditions that do not belong in any of the above two categories must correspond to a closed periodic trajectory in the $(u,v)$ phase plane. Since inter-vortex distances, $r_1$, $r_2$, and $r_{12}$ are functions of $u$ and $v$ [see~\eqref{Mneq0_r1_r2_r12}], they will also be periodic functions of time. This is illustrated by considering an initial conditions $z_1|_{t=0}=1$ and $z_2|_{t=0}=0.5$ and integrating the system~\eqref{z0_z1_z2_dot} numerically till $t=21$. As expected, the $(u,v)$ phase plane trajectory is a closed orbit [see figure~\ref{fig_Mneq0_equal_periodic}(a)], and $r_1,r_2,r_{12}$ are periodic [see figure~\ref{fig_Mneq0_equal_periodic}(c)], resulting in a vortex motion as in figure~\ref{fig_Mneq0_equal_periodic}(b). 

\begin{figure}[htb!]
\centering
\subfigure[]{
\includegraphics[width=0.31\textwidth,height=0.31\textwidth]{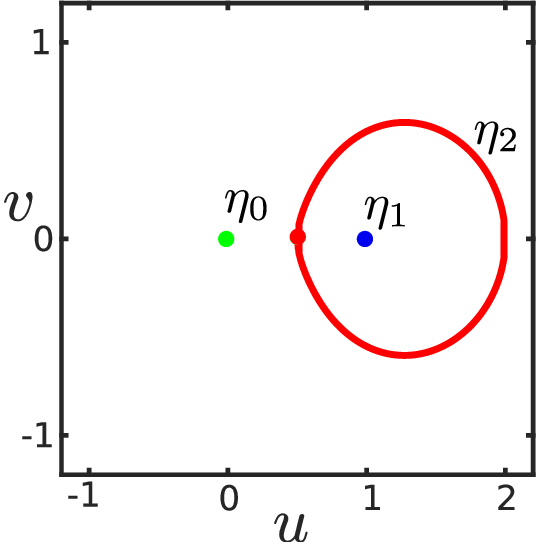}
}
\subfigure[]{
\includegraphics[width=0.31\textwidth,height=0.31\textwidth]{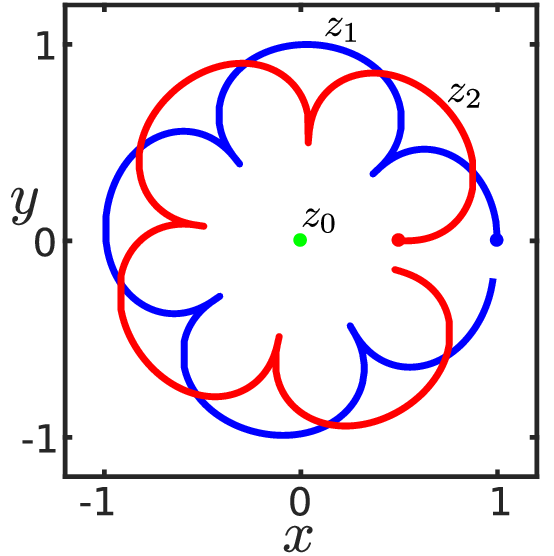}
}
\subfigure[]{
\includegraphics[width=0.31\textwidth,height=0.31\textwidth]{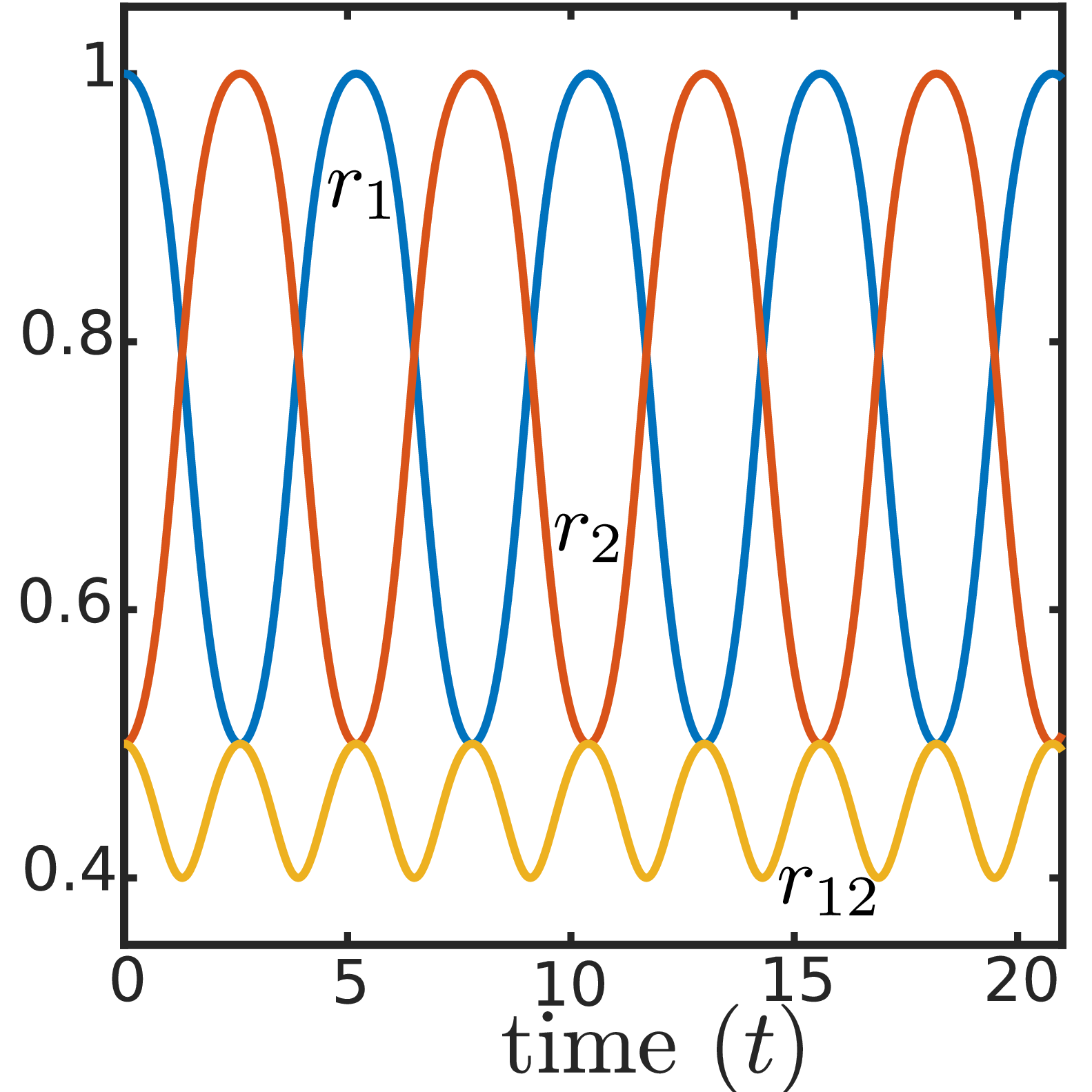}
}
\caption{\small{
Same as figure~\ref{fig_Mneq0_equal_fixed_conf} but for the periodic case.
}}
\label{fig_Mneq0_equal_periodic}
\end{figure}

\end{enumerate}

\subsubsection{Equal counter-rotating pair (\texorpdfstring{$\Gamma_1+\Gamma_2=0$}{Gamma1+Gamma2=0})}

This case is of particular interest because, unlike other cases, the vortex motions are not bounded generally. For some initial conditions, free vortex pair gets entrapped to a neighbourhood of the fixed vortex, and for some they escape to infinity. The existence of a boundary that separates the former from the latter is explained through the examples below. 
WLOG, we may assume that initial conditions for $z_1$ and $z_2$ are such that $|z_2/z_1|_{t=0}<1$ (see remark~\ref{remark_sign_convention}). We shall consider two sets of circulations to illustrate the situations when the unique saddle on the $u$-axis lies in the intervals (i) $(0,1)$ and (ii) $(-1,0)$.
 
 \begin{figure}[htb!]
\centering
\subfigure[]{
\includegraphics[width=0.35\textwidth,height=0.35\textwidth]{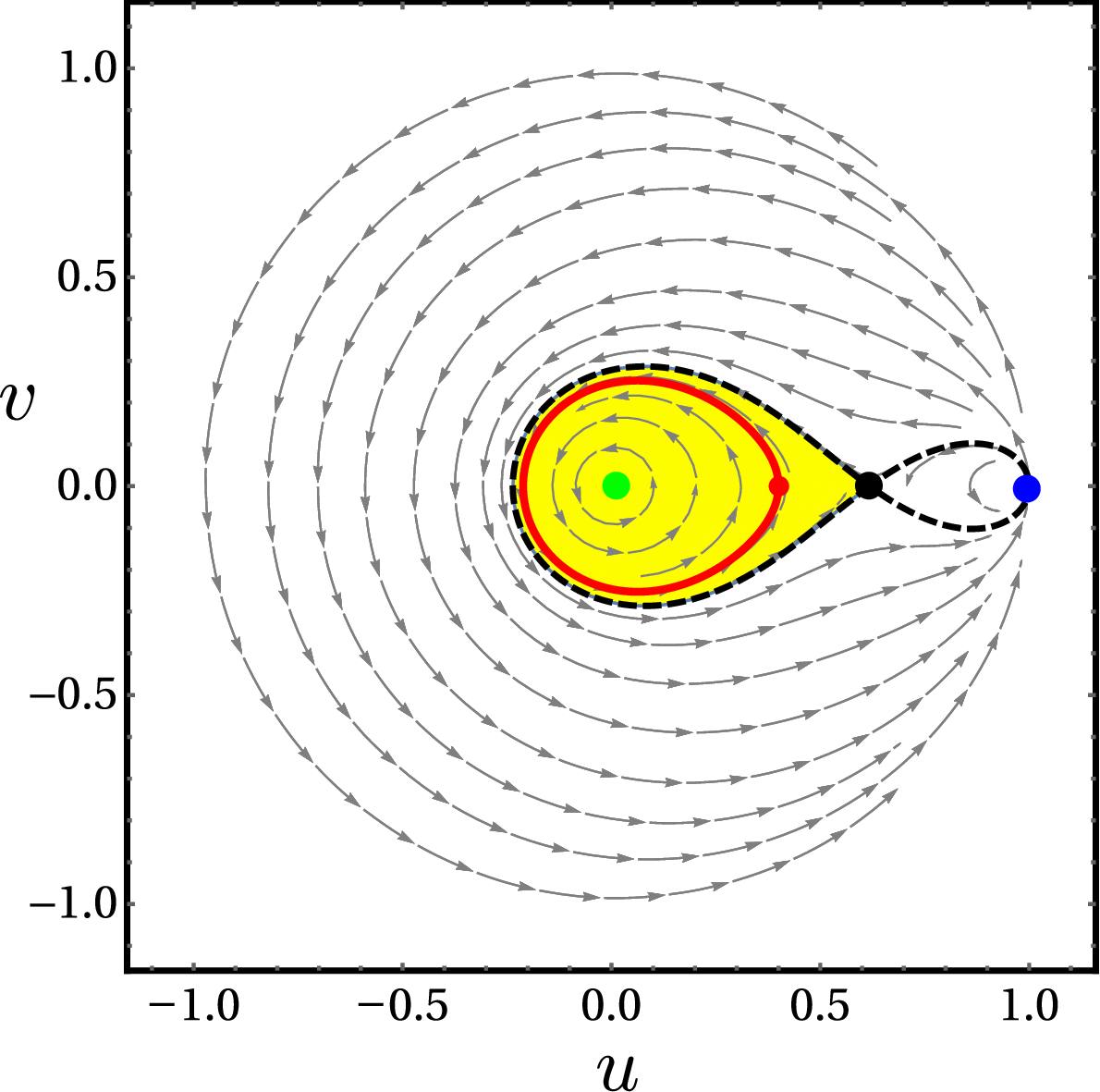}}\qquad\qquad
\subfigure[]{\includegraphics[width=0.35\textwidth,height=0.35\textwidth]{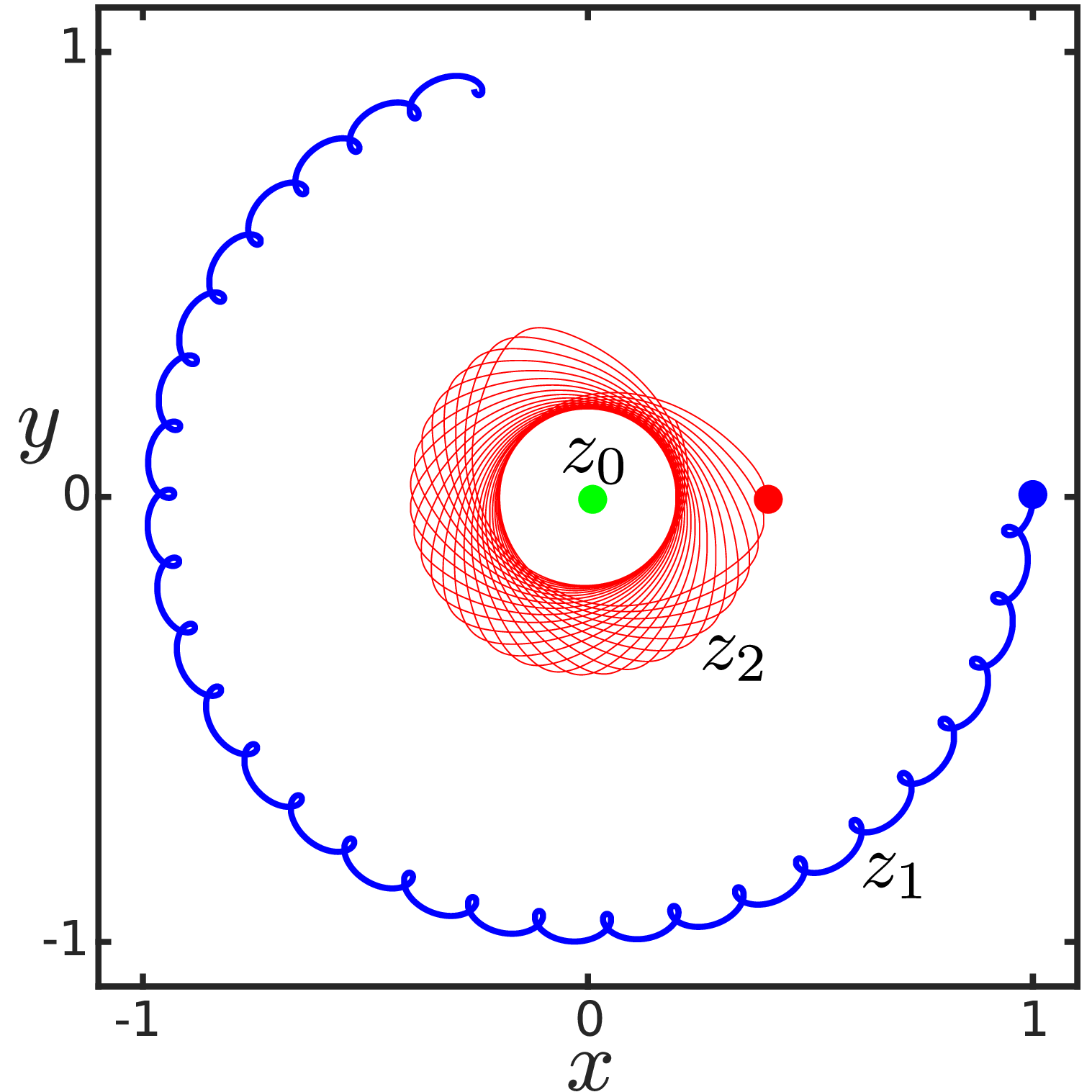}}
\caption{
\small{
An example of bounded vortex motion for the equal counter-rotating pair case ($\Gamma_0=\Gamma_1=1$, $\Gamma_2=-1$). 
The positions of vortices $\mathcal{V}_0$, $\mathcal{V}_1$ and $\mathcal{V}_2$ are marked by green, blue, and red, respectively, in
%
(a) $(u,v)$ phase plane, and 
(b) $(x,y)$ plane. The black dot in panel (a) represents the saddle.}}
\label{fig_Mneq0_Keq1_g0_p_bounded}
\end{figure}
 
Let us look at the case when the saddle point  on the $u$-axis lies in $(0,1)$. We have considered the circulations as $\Gamma_0=1,\Gamma_1=1,\Gamma_2=-1$, so that the unique saddle is at $((\sqrt{5}-1)/2,0)\approx(0.618034,0)$. As explained in theorem~\ref{theorem_Mneq0_ultimatum}, separatrices [see black dashed lines in figures~\ref{fig_Mneq0_Keq1_g0_p_bounded}(a) and~\ref{fig_Mneq0_Keq1_g0_p_unbounded}(a)] divide the $(u,v)$ phase plane into three sub-regions. The region that contains the origin is shaded yellow.

\begin{figure}[htbp!]
\centering
\subfigure[]{
\includegraphics[width=0.35\textwidth,height=0.35\textwidth]{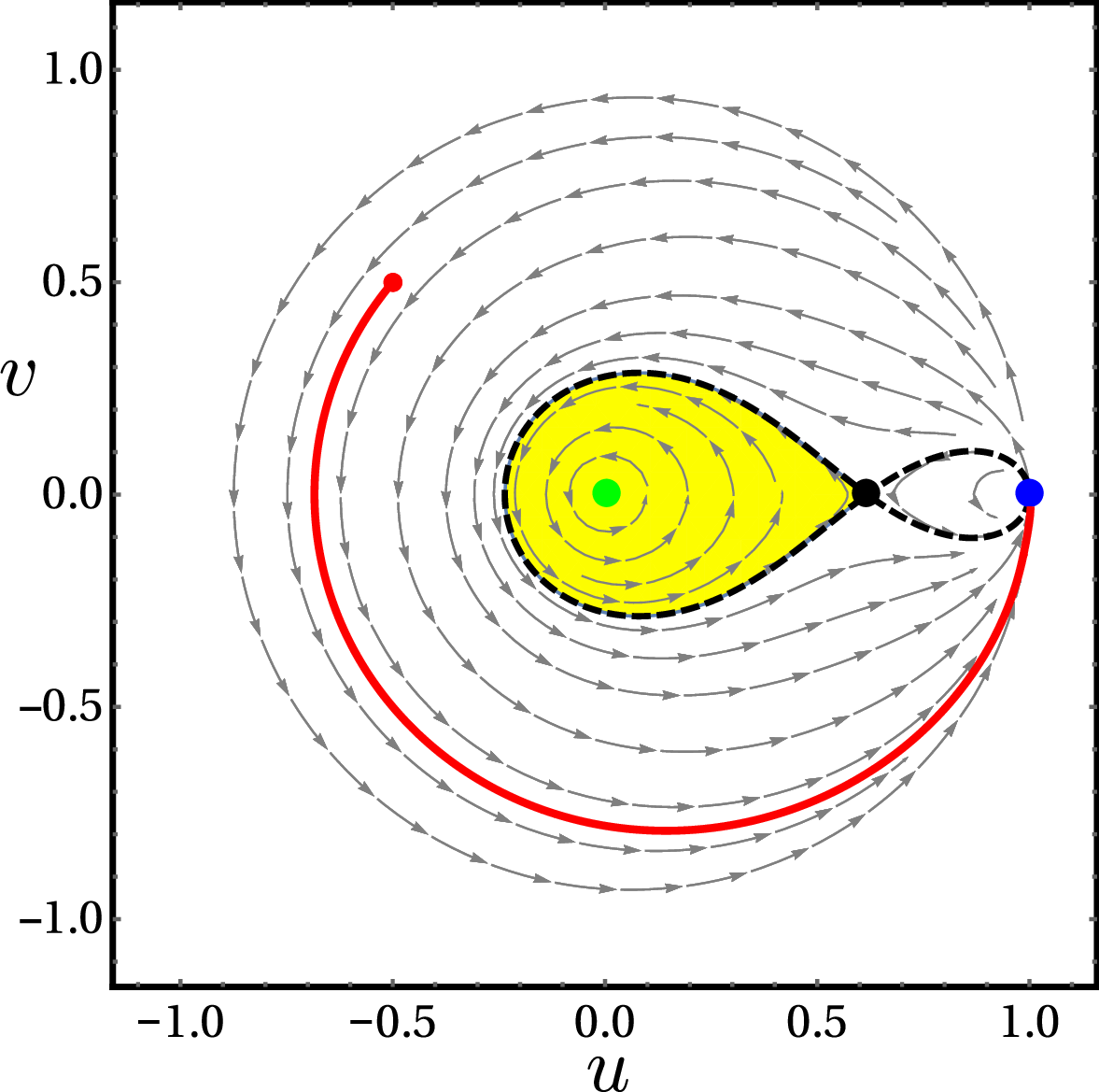}}\qquad\qquad
\subfigure[]{\includegraphics[width=0.35\textwidth,,height=0.35\textwidth]{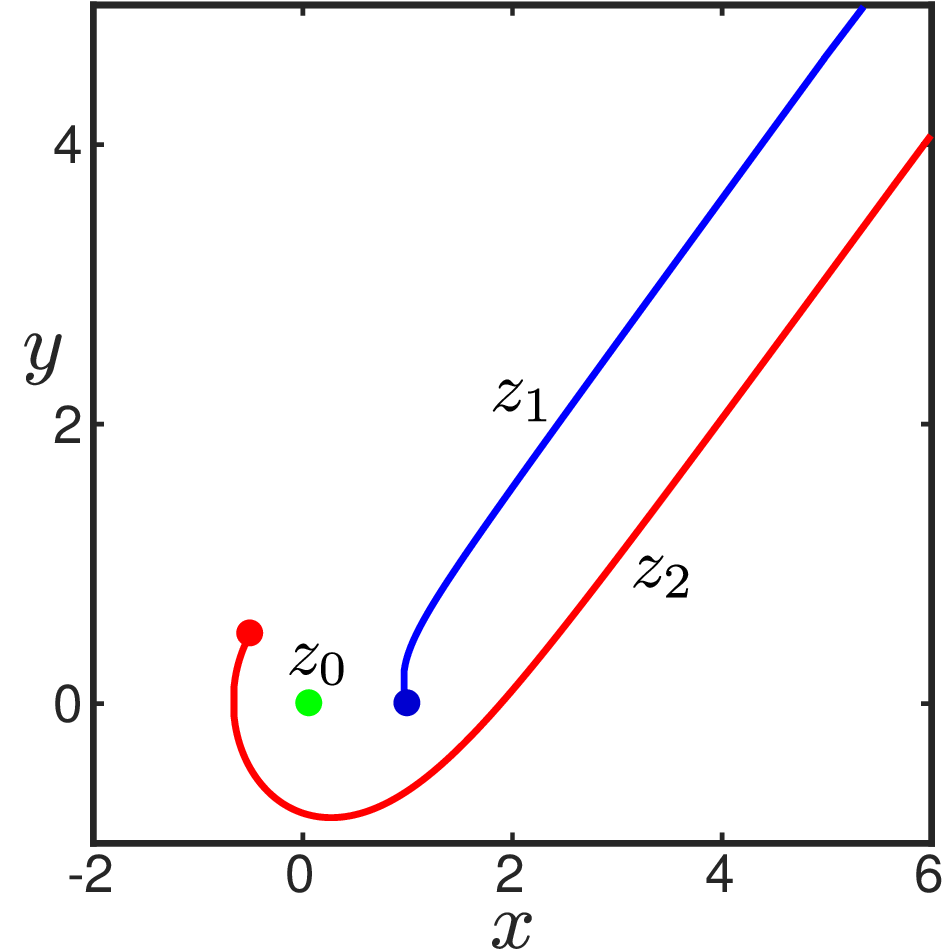}}
\caption{
\small{
Same as figure~\ref{fig_Mneq0_Keq1_g0_p_bounded} but for the unbounded vortex motion.}
}
\label{fig_Mneq0_Keq1_g0_p_unbounded}
\end{figure}

An initial condition for $z_1$ and $z_2$ is arbitrarily chosen such that the ratio $z_2/z_1$ lies in this region. In figure~\ref{fig_Mneq0_Keq1_g0_p_bounded}, we have taken the initial conditions as $z_2|_{t=0}=0.4$, $z_1|_{t=0}=1$
[marked by a red dot in figure~\ref{fig_Mneq0_Keq1_g0_p_bounded}(a)] that lies in the yellow shaded region. For these initial conditions, system~\eqref{z0_z1_z2_dot} is numerically integrated to obtain the $(u,v)$ phase plane trajectory [marked red in figure~\ref{fig_Mneq0_Keq1_g0_p_bounded}(a)] as well as the actual vortex trajectories [see figure~\ref{fig_Mneq0_Keq1_g0_p_bounded}(b)]. As expected from theorems~\ref{theorem_Mneq0_periodic_orbits} and~\ref{theorem_Mneq0_ultimatum}, the $(u,v)$ phase plane trajectory is a closed orbit, and the vortex motion is bounded with periodic inter-vortex distances 
(figure not shown for brevity).

Next, we illustrate the case when the initial condition is such that $z_2/z_1|_{t=0}$ is outside the region of entrapment. We have considered the initial conditions $z_2|_{t=0}=-0.5+ 0.5\, \mathbbm{i}$ and $z_1|_{t=0}=1$, so that $z_2/z_1|_{t=0}=-0.5+ 0.5\, \mathbbm{i}$ [marked as a red dot in figure~\ref{fig_Mneq0_Keq1_g0_p_unbounded}(a)], lies outside the region of entrapment as required. By numerically plotting the respective $(u,v)$ phase plane trajectory [marked red in figure~\ref{fig_Mneq0_Keq1_g0_p_unbounded}(a)] and the physical vortex trajectories [see figure~\ref{fig_Mneq0_Keq1_g0_p_unbounded}(b)], we see that the vortex motion is unbounded and the corresponding $(u,v)$ phase plane trajectory tends to the singularity point $(1,0)$ just as one would expect from lemma~\ref{lemma_Mneq0_eta1_convergence} and theorem~\ref{theorem_Mneq0_ultimatum}. 

\begin{figure}[htb!]
\centering
\subfigure[]{
\includegraphics[width=0.35\textwidth,height=0.35\textwidth]{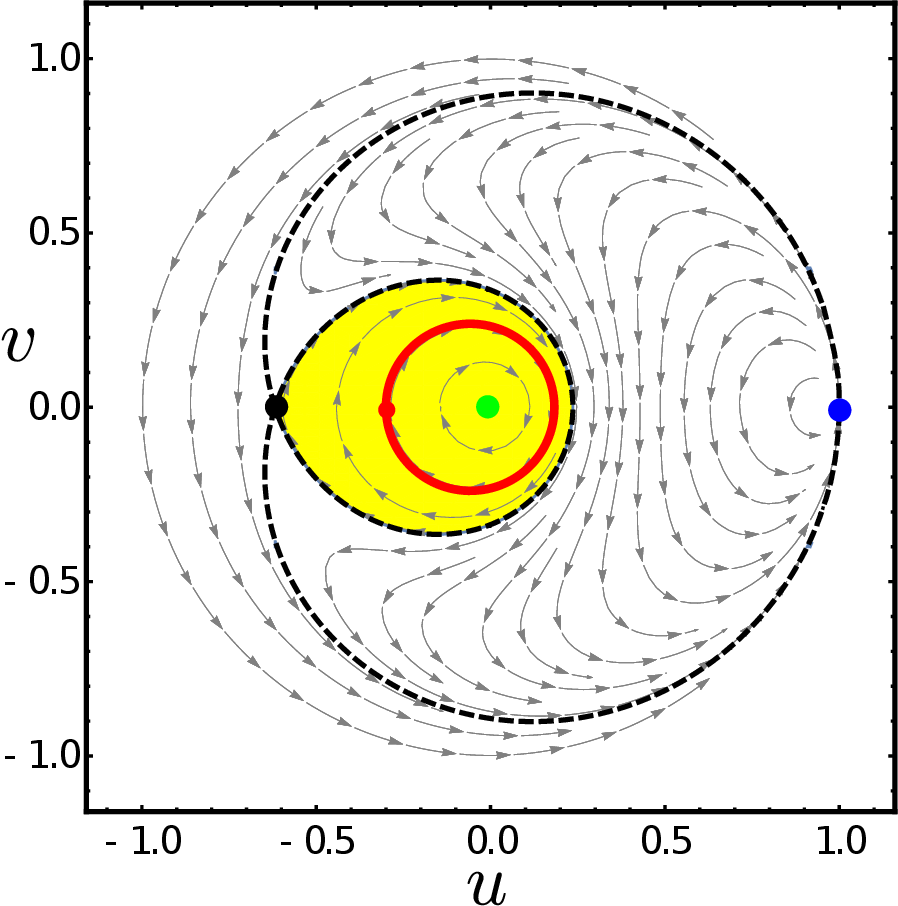}}\qquad\qquad
\subfigure[]{\includegraphics[width=0.35\textwidth,height=0.35\textwidth]{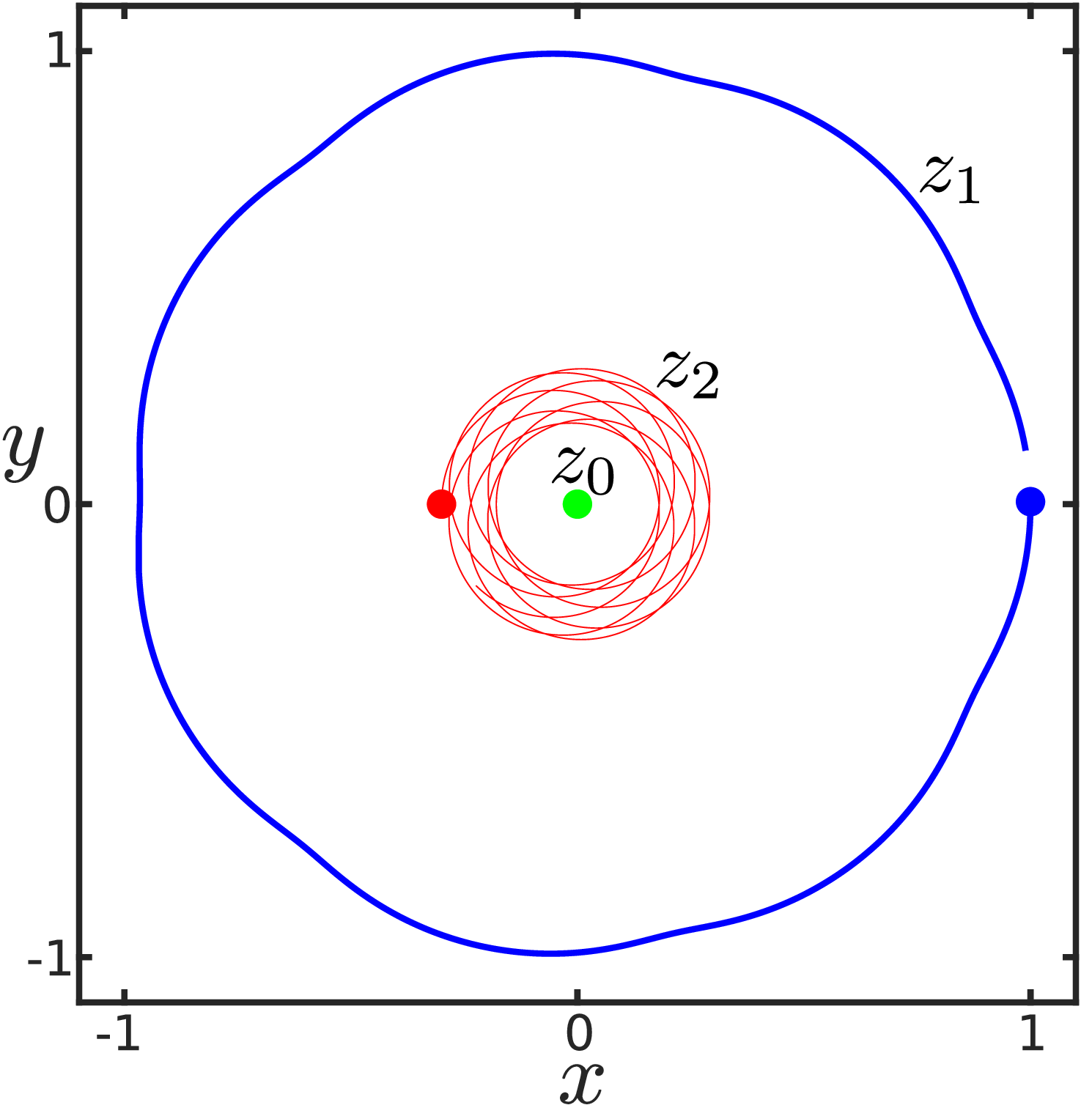}}
\caption{\small{
Same as figure~\ref{fig_Mneq0_Keq1_g0_p_bounded} but for a different set of circulations $\Gamma_0=\Gamma_2=-1$, $\Gamma_1=1$. }}
\label{fig_Mneq0_Keq1_g0_m_bounded}
\end{figure}

To illustrate the case when the unique saddle lies in $(-1,0)$ interval on the $u$-axis, we consider the circulations to be $\Gamma_0=-1,\Gamma_1=1,\Gamma_2=-1$, so that the saddle is at $((1-\sqrt{5})/2,0)\approx (-0.618034,0)$. The separatrices divide the $(u,v)$ phase plane into three, as seen in figures~\ref{fig_Mneq0_Keq1_g0_m_bounded}(a) and~\ref{fig_Mneq0_Keq1_g0_m_unbounded}(a). The region that contains the origin (shaded yellow) is the region of vortex entrapment, as given by theorem~\ref{theorem_Mneq0_ultimatum}. For the initial conditions $z_1|_{t=0}=1$ and $z_2|_{t=0}=-0.3$, the ratio $z_2/z_1|_{t=0}$ [red dot in figure~\ref{fig_Mneq0_Keq1_g0_m_bounded}(a)] lies in the region of entrapment. Vortex trajectories as obtained from numerically integrating the system~\eqref{z0_z1_z2_dot} for $0\leq t\leq20$ clearly shows that the vortex motion is bounded with periodic inter-vortex distance functions [see figure~\ref{fig_Mneq0_Keq1_g0_m_bounded}(b)].

For the initial conditions $z_1|_{t=0}=1$ and $z_2|_{t=0}=0.5\, \mathbbm{i}$, the ratio $z_2/z_1|_{t=0}$ [red dot in figure~\ref{fig_Mneq0_Keq1_g0_m_unbounded}(a)] lies outside the region of entrapment and as expected the corresponding  vortex motion is found to be unbounded [see figure~\ref{fig_Mneq0_Keq1_g0_m_unbounded}(b)].
\begin{figure}[htb!]
\centering
\subfigure[]{
\includegraphics[width=0.35\textwidth,height=0.35\textwidth]{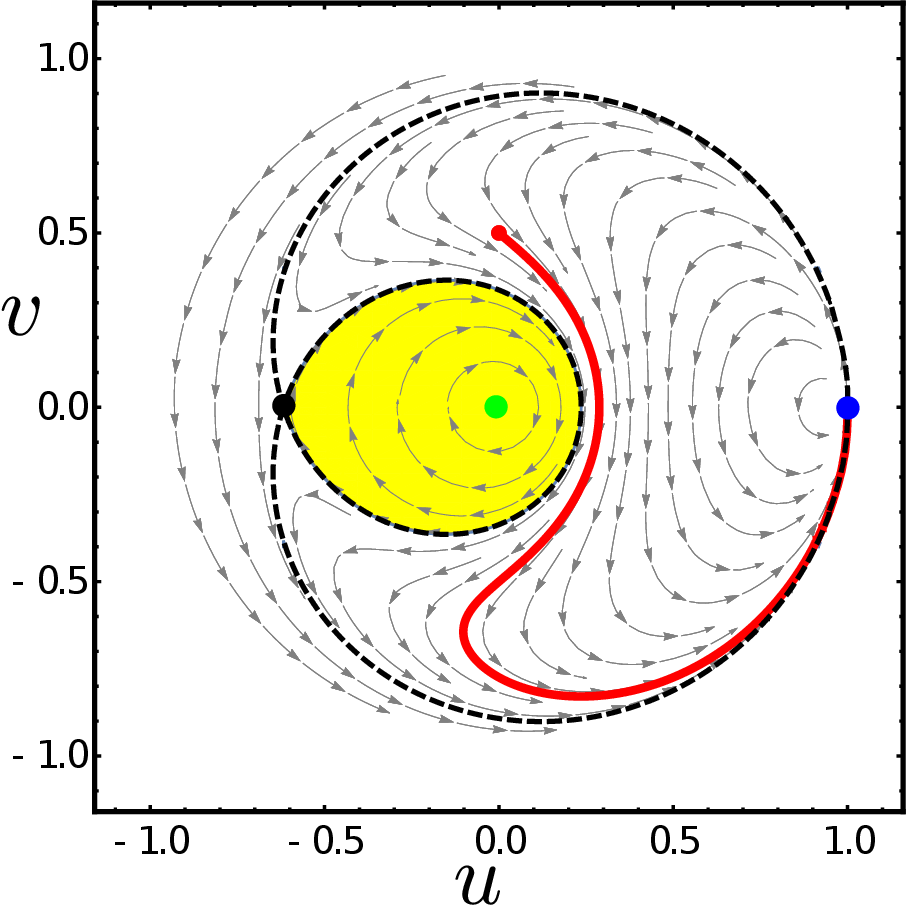}}\qquad\qquad
\subfigure[]{\includegraphics[width=0.35\textwidth,height=0.35\textwidth]{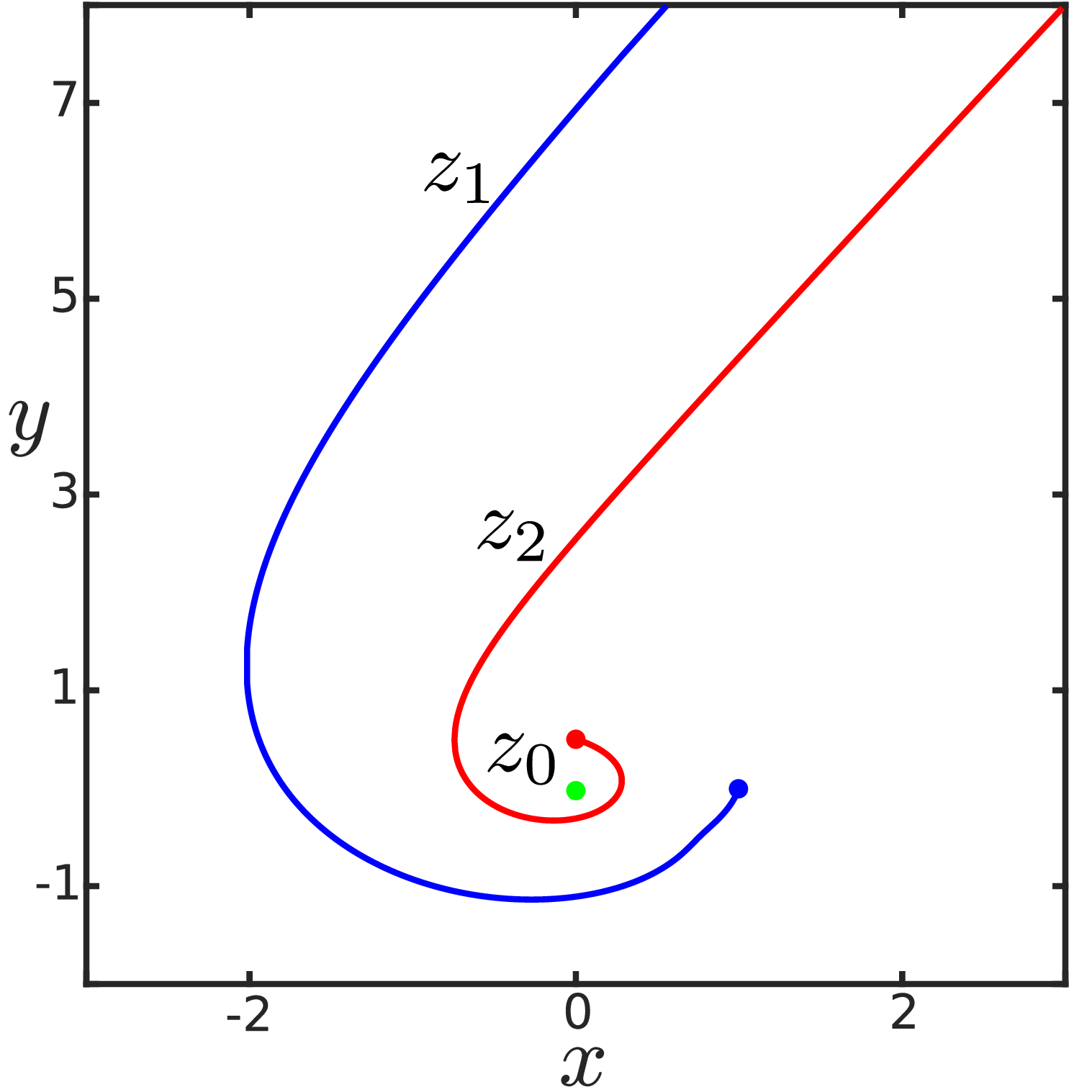}}
\caption{
\small{
Same as figure~\ref{fig_Mneq0_Keq1_g0_p_unbounded} but for a different set of circulations $\Gamma_0=\Gamma_2=-1$, $\Gamma_1=1$. 
}
}
\label{fig_Mneq0_Keq1_g0_m_unbounded}
\end{figure}

\section{Conclusions}
\label{sec:conclusions}

The dynamical aspects of a constrained three-vortex problem, in particular, different types of motion exhibited by a pair of point 
vortices $\mathcal{V}_1$ and $\mathcal{V}_2$ with circulations $\Gamma_1$ and $\Gamma_2$ in the presence of a fixed point vortex $\mathcal{V}_0$ with circulation $\Gamma_0$, where the circulations 
take arbitrary non-zero values, have been studied in detail.
Instead of directly looking at the dynamics 
based on the positions $z_1$ and $z_2$ of the free vortices, we have looked at the quotient $z_2/z_1$ to gain insights about the vortex system. The main advantage of this choice is the reduction in the number of coordinates, which simplifies the analysis. 
Depending on the value of the constant $M=\Gamma_1|z_1|^2+\Gamma_2|z_2|^2$, the problem has been classified into two cases $M=0$ and $M\neq 0$. Both these cases are illustrated in a flow
chart, see figure~\ref{fig:flowchart}, which covers all the trajectories discussed in the paper.

\begin{figure}[!th]
\begin{center}
{\includegraphics[width=\textwidth]{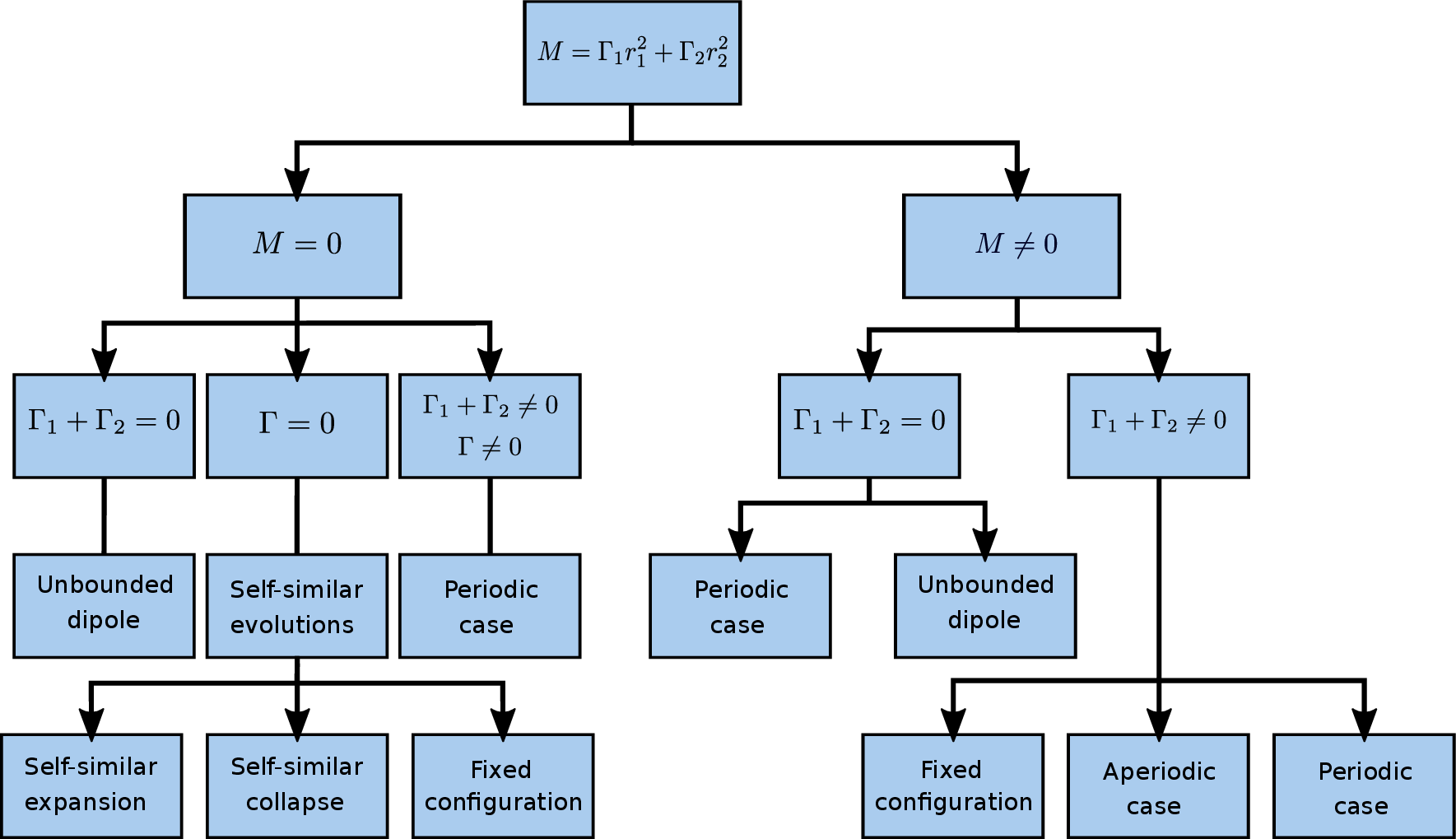}}
\end{center}
\caption{\small Flow chart summarizing all the trjectories of constrained three-vortex problem. The case $\Gamma_1 + \Gamma_2=0$ (equal counter rotating vortex pair) has been discussed in Refs.~\cite{RK2013,RE2014,KR2018}.} 
\label{fig:flowchart}
\end{figure}

For $M=0$ case, the present results show that irrespective of the initial conditions
there are three kinds of possible vortex motions depending on the value of $\Gamma=\Gamma_1\Gamma_2+\Gamma_0\Gamma_1+\Gamma_0\Gamma_2$ and $\Gamma_1+\Gamma_2$. They are (i) self-similar evolution ($\Gamma=0$), (ii) unbounded dipole motion $(\Gamma_1+\Gamma_2=0)$, and (iii) bounded periodic motion ($\Gamma_1+\Gamma_2\neq0,\Gamma\neq 0$). 
The self-similar evolutions have been further classified into 
three types: self-similar expansion, self-similar collapse, and fixed configuration based on
the distance from the free vortices to the fixed vortex that increases, decreases, and remains constant with respect to time, respectively.

For the equal counter-rotating vortex case, i.e.,~$\Gamma_1+\Gamma_2=0$, we find that the two free vortices $\mathcal{V}_1$ and $\mathcal{V}_2$ always escape to infinity. In other situations, i.e.,~$\Gamma_1+\Gamma_2\neq 0$ and $\Gamma\neq 0$, 
we notice that the free vortices are bounded in a neighbourhood of the fixed vortex with periodic inter-vortex distances.

For $M\neq 0$, we confirm that there are no self-similar expansions or collisions, which contrasts markedly with the vortex motion in the case of $M=0$. In general for $M\neq 0$, we establish that a vortex motion can have one of the form:~(i) a fixed configuration, where vortices move in circular orbits 
around the fixed vortex in a collinear fashion, (ii) a bounded motion, where the free vortices asymptotically approach to a fixed configuration, (iii) a bounded vortex motion, where inter-vortex distances are periodic, and vortices oscillate between two distinct collinear configurations, and  (iv) an unbounded vortex motion. 
Our analysis also elaborates
that for an unbounded vortex motion, it is necessary that the free vortices being the equal counter-rotating pair, and irrespective of vortex circulations and initial conditions, the distance between the free vortices, $\mathcal{V}_1$ and $\mathcal{V}_2$, remains bounded from both sides. Furthermore, for the equal counter-rotating case, the necessary and sufficient condition for a vortex entrapment is that the initial quotient, i.e.,~$z_2/z_1|_{t=0}$, remains in the interior of the curve given by $\Psi(u,v)=\Psi(u_s,0)$ that encloses the origin, where $\Psi$ is given by~\eqref{Mneq0_levelcurve}, and $u_s$ is the unique real root of the polynomial $u^2-(\Gamma_2/\Gamma_0) u-1$ within $(-1,1)$. 
The present analysis depends mainly on the tools of the dynamical system.

\section*{Acknowledgment}

The authors sincerely thank the anonymous referees
for their 
constructive comments which helped improve the paper. 
The authors thank Prof.~Shaligram Tiwari, IIT Madras, for motivating discussions at the early stages of this work.  
P.S. acknowledges financial support from IIT Madras through the Grant No. MAT/16-17/671/NFSC/PRIY.










\appendix
\bibliographystyle{rspublicnatwithsort_implicitdoi}
\bibliography{references}

\end{document}